\documentclass[11pt]{article}

\usepackage[margin=1in]{geometry}

\usepackage{times}

\usepackage{soul}
\usepackage{url}
\usepackage[hidelinks]{hyperref}
\usepackage[utf8]{inputenc}
\usepackage[small]{caption}
\usepackage{graphicx}
\usepackage{amsmath,amssymb,amsthm}
\usepackage{amsfonts}
\usepackage{booktabs}
\usepackage{color}
\usepackage{xfrac}
\usepackage{tikz}
\usepackage{pgfplots}
\pgfplotsset{compat=1.17}
\usetikzlibrary{patterns}
\usepackage{pgfplotstable}
    \pgfplotsset{
    name nodes near coords/.style={
        every node near coord/.append style={
            name=#1-\coordindex,
            alias=#1-last,
        },
    },
    name nodes near coords/.default=coordnode
    }

\newcommand{\bX}{\mathbf{X}}
\newcommand{\bA}{\mathbf{A}}

\usepackage[ruled,linesnumbered,vlined]{algorithm2e}
\usepackage[capitalize]{cleveref}

\newenvironment{proofof}[1]{{\vspace*{5pt} \noindent\bf Proof of #1:  }}{\hfill\rule{2mm}{2mm}\vspace*{5pt}}

\newtheorem{claim}{Claim}[section]
\newtheorem{theorem}{Theorem}[section]

\newtheorem{lemma}[theorem]{Lemma}
\newtheorem{observation}[theorem]{Observation}

\newtheorem{definition}[theorem]{Definition}

\DeclareMathOperator*{\argmax}{argmax}
\DeclareMathOperator*{\argmin}{argmin}

\title{Approximately EFX Allocations for Indivisible Chores}

\author{Shengwei Zhou
	\thanks{IOTSC, University of Macau. \{yc17423,xiaoweiwu\}@um.edu.mo. This work is funded by the Science and Technology Development Fund (FDCT), Macau SAR (file no. 0014/2022/AFJ, 0085/2022/A, 0143/2020/A3 and SKL-IOTSC-2021-2023).}
	\and
	Xiaowei Wu $^*$
}

\date{}

\begin{document}
	\pagestyle{plain}
	
	\maketitle

        \begin{abstract}
        	In this paper, we study how to fairly allocate a set of $m$ indivisible chores to a group of $n$ agents, each of which has a general additive cost function on the items.
        	Since envy-free (EF) allocations are not guaranteed to exist, we consider the notion of envy-freeness up to any item (EFX).
        	In contrast to the fruitful results regarding the (approximation of) EFX allocations for goods, very little is known for the allocation of chores.
        	Prior to our work, for the allocation of chores, it is known that EFX allocations always exist for two agents or general number of agents with identical ordering cost functions.
        	For general instances, no non-trivial approximation result regarding EFX allocation is known.
        	In this paper, we make progress in this direction by providing several polynomial time algorithms for the computation of EFX and approximately EFX allocations.
            We show that for three agents we can always compute a $(2+\sqrt{6})$-approximation of EFX allocation.
        	For $n\geq 4$ agents, our algorithm always computes a $(3n^2-n)$-approximation.
        	We also study the bi-valued instances, in which agents have at most two cost values on the chores.
            For three agents, we provide an algorithm for the computation of EFX allocations.
            For $n\geq 4$ agents, we present algorithms for the computation of partial EFX allocations with at most $(n-1)$ unallocated items; and $(n-1)$-approximation of EFX allocations.
        \end{abstract}

	\section{Introduction}
	
	Fairness is receiving increasing attention in a broad range of research fields, including but not limited to computer science, economics, and mathematics.
	A fair allocation problem focuses on allocating a set $M$ of $m$ items to a group $N$ of $n$ agents, where different agents may have different valuation functions on the items.
	When the valuation functions give positive values, the items are considered as goods, e.g., resources; when the valuation functions give negative values, the items are considered as chores, e.g., tasks.
	In the latter case, we refer to the valuation functions as cost functions.
	In this paper, we focus on the situation when the functions are additive.
	Arguably, two of the most well-studied fairness notions are {\em envy-freeness} (EF) \cite{foley1967resource} and {\em proportionality} (PROP) \cite{steihaus1948problem}.
	Proportionality means that each agent received at least her proportional share of all items.
	Envy-freeness is even stronger.
	Informally speaking, an allocation is EF if no agent wants to exchange her bundle of items with another agent in order to increase her utility.
	In contrast to the case of divisible items, where EF and PROP allocations always exist
	\cite{alon1987splitting,edward1999rental,conf/focs/AzizM16}, when items are indivisible, they are not guaranteed to exist even for some simple cases.
	For example, consider allocating a single indivisible item to two agents.
	This example also defies approximations of EF and PROP allocation. 
	Therefore, researchers have turned their attention to relaxations of these fairness notions.
	{\em Envy-freeness up to one item} (EF1) \cite{conf/sigecom/LiptonMMS04} and {\em envy-freeness up to any item} (EFX) \cite{journals/teco/CaragiannisKMPS19} are two widely studied relaxations of EF.
	Informally speaking, an EF1 allocation requires that the envy between any two agents can be eliminated by removing some item; while an EFX allocation requires that the envy can be eliminated by removing any item.
 
	It has been shown that EF1 allocations are guaranteed to exist and can be found in polynomial time for goods~\cite{conf/sigecom/LiptonMMS04}, chores and even mixture of the two~\cite{conf/approx/BhaskarSV21}.
	However, EF1 could sometimes lead to extreme unfairness, even if a much fairer allocation exists.
	EFX, on the other hand, puts a much stronger constraint on the allocation and is arguably the most compelling fairness notion.
	There are fruitful results regarding the existence and computation of (approximations of) EFX allocations since the notion was first proposed by Caragiannis et al.~\cite{journals/teco/CaragiannisKMPS19}.
	{  For the allocation of goods, it has been shown that EFX allocations exist for two agents with general valuations and any number of agents with identical ordering (IDO) valuations~\cite{journals/siamdm/PlautR20}, and three agents with additive valuations~\cite{conf/sigecom/ChaudhuryGM20}.
    Very recently, Akrami et al.~\cite{conf/sigecom/AkramiACGMM23} generalized the latter result to three agents with nearly general valuations.}
	It remains a fascinating open problem whether EFX allocations always exist in general.
	For general number of agents with additive valuations, there are efficient algorithms for the computation of  $0.5$-approximate\footnote{Regarding approximations of EFX allocations, the approximate ratios are at most $1$ for goods, and at least $1$ for chores.} EFX allocations~\cite{journals/siamdm/PlautR20,conf/atal/Chan0LW19} and $0.618$-approximate EFX allocations~\cite{journals/tcs/AmanatidisMN20}.
  
	In contrast to the allocation of goods, very little regarding EFX allocations for chores is known to this day.
	It can be easily shown that the divide-and-choose algorithm computes an EFX allocation {  for} two agents with general valuations. 
	{  Recently, it has been shown that EFX allocations always exist for some special cases, e.g., IDO instances~\cite{conf/www/LiLW22}, binary instances~\cite{journals/corr/abs-2308-12177}, instances with two types of chores~\cite{conf/atal/0001LRS23}, and instances with leveled preferences~\cite{journals/corr/abs-2109-08671}.}
	However, even for three agents with general additive valuations, it is unknown whether constant approximations of EFX allocations exist, let alone the existence of EFX allocations.
	
	\subsection{Main Results}
	
	In this paper, we propose polynomial-time algorithms for the computation of approximately EFX allocations for indivisible chores.
	For three agents, our algorithm achieves an approximation ratio of $2+\sqrt{6}$ while for $n\geq 4$ agents the approximation ratio is $3n^2-n$.
	Prior to our work, no non-trivial results regarding the approximation of EFX allocation for chores are known, except for some special cases~\cite{conf/www/LiLW22,journals/corr/abs-2109-08671}.

	\smallskip
	\noindent
	{\bf Result 1} (Theorem \ref{thm:n=3}){\bf .}
	{\em There exists a polynomial time algorithm that computes a $(2+\sqrt{6})$-approximate EFX allocation for three agents with additive cost functions.}
	
	\smallskip

	\noindent
	{\bf Result 2} (Theorem \ref{thm:general-n}){\bf .}
	{\em There exists a polynomial time algorithm that computes a $(3n^2-n)$-approximate EFX allocation for $n\geq 4$ agents with additive cost functions.}

	\paragraph{Main Challenge}
	While being two seemingly similar problems, the EFX allocations of goods and chores admit distinct difficulties in approximability.
	For the allocation of goods, while computing {  EFX allocations is} difficult, getting a constant approximation ratio turns out to be quite straightforward.
	Existing algorithms~\cite{journals/tcs/AmanatidisMN20,journals/siamdm/PlautR20} for the computation of approximately EFX allocations for goods are mainly based on the 
	\emph{Envy-Cycle Elimination} technique~\cite{conf/sigecom/LiptonMMS04,conf/approx/BhaskarSV21}.
	Roughly speaking, in these algorithms, in each round an ``un-envied'' agent will be chosen to pick her favourite unallocated item, where the envy-cycle elimination technique ensures that there will always be an un-envied agent in each round.
	The result  of~\cite{conf/www/LiLW22} also follows this framework to compute an EFX allocation for chores on identical ordering instances.
	The key to the analysis is showing that the value/cost of an agent increases by at most a small constant factor in each round.
	However, it seems quite challenging to extend this framework to handle general instances for the allocation of chores.
	For the allocation of goods, it can be shown that the utilities of agents are non-decreasing in each round: the picking operation and envy-cycle elimination do not decrease the value of any agent.
	In contrast, for the allocation of chores, when an agent picks an item, its cost increases; when an envy-cycle is eliminated, the cost of the involved agents decreases.
	This introduces a main difficulty in computing an (approximation of) EFX allocation for chores: if the cost of an agent is very small when it picks an item with large cost, the approximation ratio of EFX can be arbitrarily bad.
	
	\paragraph{Our Techniques}
	To get around this difficulty, we adopt a completely different approach in {  allocating} the items.
	Our first observation is that when all items have small costs to all agents, then there is likely to exist a partition that looks ``even'' to all agents.
	To handle this case, we propose the Sequential Placement algorithm, which computes a partition of the items for a group of agents such that the ratios between the cost of any two bundles are bounded, under every agent's cost function.
	On the other hand, if there exists an item $e$ that has large cost to an agent $i$, then by allocating the item to some other agent $j$, we can ensure that agent $i$ does not strongly envy agent $j$, no matter what items agent $i$ received eventually.
	Our algorithms rely on a careful combination of the above two observations.
	For three agents, we show that by classifying the agents into different types depending on the number of large items they have, we are able to get a $(2+\sqrt{6})$-approximate EFX allocation.
	To extend the ideas to general number agents, we borrow some existing techniques for the computation of PROP1 allocation for goods~\cite{conf/sigecom/LiptonMMS04}, and bound the approximation ratio by $3n^2-n$, where $n$ is the number of agents.
	
	\medskip
	
	We also show that our results can be improved for \emph{bi-valued} instances.
	An instance is called a bi-valued if there exist two values $a,b\geq 0$ such that $c_i(e)\in\{a,b\}$ for every agent $i\in N$ and item $e\in M$.
	In other words, each agent classifies the set of items into \emph{large} items and \emph{small} items, and items in the same category have the same cost to the agent.
	The bi-valued instances are commonly considered as one of the most important special cases of the fair allocation problem~\cite{journals/corr/abs-2006-15747,conf/aaai/AkramiC0MSSVVW22,conf/sagt/GargM21,journals/tcs/AmanatidisBFHV21,conf/aaai/GargMQ22,conf/atal/EbadianP022}.
	For the case of bi-valued goods, polynomial time algorithms have been proposed for the computation of EFX allocations~\cite{journals/tcs/AmanatidisBFHV21} and EFX allocations that are also Pareto optimal (PO)~\cite{conf/sagt/GargM21}.
	{  It has been shown that for bi-valued chores, EF1 and PO allocations can be computed in polynomial time~\cite{conf/aaai/GargMQ22,conf/atal/EbadianP022}.}
	Unfortunately, no non-trivial result regarding the approximation of EFX allocations is known for bi-valued chores.
	
	In this paper, we make progress toward answering this problem for bi-valued instances.
	For three agents, we propose an algorithm that always computes an EFX allocation; for $n\geq 4$ agents, we propose 
    an algorithm that always computes an EFX partial allocation with at most $(n-1)$ unallocated items; and an algorithm that always computes an $(n-1)$-approximate EFX allocation.
    All our algorithms run in polynomial time.
 
	\smallskip
	\noindent
	{\bf Result 3} (Theorem \ref{thm:efx-3-agents-bi-valued}){\bf .}
	{\em For any instances with bi-valued cost functions, there exists a polynomial time algorithm that computes an EFX allocation for three agents.}
	
	\smallskip
	\noindent
	{\bf Result 4} (Theorem \ref{thm:partial-n-bivalued} and \ref{thm:general-n-bivalued}){\bf .}
	{\em For any instances with $n\geq 4$ agents with bi-valued cost functions, we can compute in polynomial time an EFX partial allocation with at most $(n-1)$ unallocated items and an $(n-1)$-approximate EFX allocation.
    }
 
	\subsection{Other Related Works}
	
	For the allocation of goods, there are also works that study partial EFX allocations, i.e., EFX allocations with some of the items unallocated.
	To name a few, Chaudhury et al.~\cite{conf/soda/ChaudhuryKMS20} show that EFX allocations exist if we are allowed to leave at most $n-1$ items unallocated.
	The result has been improved to $n-2$ items by Berger et al.~\cite{conf/aaai/BergerCFF22}.
	They also show that an EFX allocation that leaves at most one item unallocated exists for four agents.
	Chaudhury et al.~\cite{conf/sigecom/ChaudhuryGMMM21}, Berendsohn et al.~\cite{conf/mfcs/BerendsohnBK22}, and Akrami et al.~\cite{conf/sigecom/AkramiACGMM23} show that a $(1-\epsilon)$-approximate EFX allocation with sublinear number of unallocated goods and high Nash
	welfare exists.
	It remains unknown whether similar results (for the computation of approximately EFX partial allocation) hold for the allocation of chores.
	
	Besides EFX, there are other well-studied fairness notions, e.g., MMS~\cite{conf/bqgt/Budish10} and PROPX~\cite{journals/annurev/Herve19}, for the allocation of chores.
	{  While it has been shown that MMS allocations are not guaranteed to exist for indivisible chores~\cite{conf/aaai/AzizRSW17}, many works study the approximation of MMS~ \cite{conf/aaai/AzizRSW17,aziz2022approximate,conf/sigecom/HuangL21}, which results in the state-of-the-art approximation ratio of $13/11$~\cite{conf/sigecom/HuangS23}.}
	Regarding PROPX allocations, in contrast to the allocation of goods, {  for which} PROPX allocations are not guaranteed to exist~\cite{journals/annurev/Herve19,journals/orl/AzizMS20}, it is shown that PROPX allocations always exist and can be computed efficiently for chores~\cite{conf/www/LiLW22}.
    For a more detailed review of the algorithms and results for the fair allocation problem, please refer to the recent surveys by Aziz et al.~\cite{journals/sigecom/AzizLMW22} and Amanatidis et al.~\cite{journals/ai/AmanatidisABFLMVW23}.

	\subsection{Paper Outline}
    We first present the formal definitions and notions for the problem in Section~\ref{sec:preliminaries}.
	Then we consider the approximate EFX allocations for three agents and for four or more agents in Section~\ref{sec:three} and~\ref{sec:fourormore}, respectively.
	We consider the bi-valued instances with three agents in Section~\ref{sec:bi-three}, where we introduce the algorithm for computing EFX allocations.
	In Section~\ref{sec:bi-fourormore}, we study the bi-valued instances for four or more agents, and present the algorithms for computing the partial EFX allocations and the approximation EFX allocations.
	Finally, we conclude our results and propose some open problems in Section~\ref{sec:conclusion}.

	\section{Preliminaries}\label{sec:preliminaries}
	
	We consider how to fairly allocate a set of $m$ indivisible chores $M$ to a group of $n$ agents $N$.
    {  Unless otherwise specified, in this paper we use ``item'' to refer to a chore.}
	A bundle is a subset of items $X\subseteq M$.
	An allocation is represented by an $n$-partition $\bX = (X_1,\cdots,X_n)$ of the items, where $X_i \cap X_j = \emptyset$ for all $i \neq j$ and $\cup_{i\in N} X_i = M$.
	In the allocation $\bX$, agent $i\in N$ receives bundle $X_i$.
    We call the allocation $\bX$ a \emph{partial allocation} if $\cup_{i\in N} X_i\subsetneq M$.
	Each agent $i \in N$ has an additive cost function $c_i: 2^M \rightarrow \mathbb{R}^+ \cup \{0\}$. 
	That is, for any $i\in N$ and $X \subseteq M$, $c_i(X) = \sum_{e\in X} c_i(\{e\})$.
	When there is no confusion, we use $c_i(e)$ to denote $c_i(\{e\})$ for convenience. 
    Without loss of generality, we assume that all cost functions are normalized.
	That is, for any $i\in N$, $c_i(M)=1$.
	Further, given any set $X\subseteq M$ and $e\in M$, we use $X+e$ and $X-e$ to denote $X\cup\{e\}$ and $X\setminus\{e\}$, respectively.
	
	\begin{definition}[EF]
		An allocation $\bX$ is called envy-free (EF) if $c_i(X_i) \leq c_i(X_j)$ for any $i,j\in N$.
	\end{definition}
	
	\begin{definition}[$\alpha$-EFX]
		For any $\alpha \geq 1$, an allocation $\bX$ is $\alpha$-approximate envy-free up to any item ($\alpha$-EFX) if for any $i,j \in N$ and any $e \in X_i$,
		\begin{equation*}
		    c_i(X_i - e)\leq \alpha\cdot c_i(X_j).
		\end{equation*}
		When $\alpha = 1$, the allocation $\bX$ is EFX.
	\end{definition}

    
	Let $\sigma_i(j)$ be the $j$-th most costly item under $c_i$ (with ties broken deterministically, e.g., by item ID).
	In other words, for every agent $i\in N$, we have
	\begin{equation*}
		c_i(\sigma_i(1)) \geq c_i(\sigma_i(2)) \geq \cdots \geq c_i(\sigma_i(m)).
	\end{equation*}
	
	For each agent $i\in N$, we define $M_i^- = \{ \sigma_i(j): j \geq n  \}$ as the set of \emph{tail items}.
	Observe that we have $|M_i^-|  = m-n+1$ and $c_i(e) \geq c_i(e')$ for all $e\notin M_i^-$ and $e'\in M_i^-$.
	With this observation, we show that for every instance there exists a simple allocation that is $(m-n)$-EFX.
	
	\begin{lemma}\label{lemma:m-n+1}
		There exists an $(m-n)$-EFX allocation for every instance with $m$ items and $n$ agents.
	\end{lemma}
	\begin{proof}
		Fix an arbitrary agent, say agent $n$, and define the allocation as follows.
		For all $i<n$, let $X_i = \{ \sigma_n(i) \}$.
		Let $X_n = M_n^-$ be the set containing the remaining items.
		Obviously, the resulting allocation is EFX for all agents $i < n$ since $|X_i| = 1$.
		By the definition of $M_n^-$, for every $e\in M_n^-$ and $i < n$, we have
		\begin{equation*}
		c_n(M_n^- - e) \leq (m-n)\cdot c_n(\sigma_n(i)) = (m-n)\cdot c_n(X_i).
		\end{equation*}
		Consequently, the allocation is $(m-n)$-EFX.
	\end{proof}
	Unfortunately the above approximation ratio is too large, especially when $m$ is large and $n$ is small.
	In the following sections, we present algorithms that compute allocations with approximations of EFX that depend only on $n$.
	In particular, our algorithm computes a $(2+\sqrt{6})$-EFX allocation when $n=3$ and $(3n^2-n)$-EFX allocation when $n\geq 4$.

	\section{Three Agents}\label{sec:three}
	
	In this section we present the algorithm that computes $(2+\sqrt{6})$-EFX allocations when $n=3$. 
	The ideas we use to compute the allocation in this part inspire our design of an approximation algorithm for general $n$.
	
	\begin{theorem}\label{thm:n=3}
		There exists an algorithm that computes an $\alpha$-EFX allocation where $\alpha = 2+\sqrt{6}\approx 4.4495$ for three agents in $O(m \log m)$ time.
	\end{theorem}
    
	{  
	For three agents we have  $M_i^- = M\setminus \{ \sigma_i(1),\sigma_i(2) \}$ for all $i\in N$.
	Observe that if there exists an agent $i\in N$ with $c_i(M_i^-) \leq \alpha \cdot c_i(\sigma_i(2))$, then by allocating $\sigma_i(1)$ and $\sigma_i(2)$ to the other two agents and $M_i^-$ to agent $i$, we end up having an $\alpha$-EFX allocation (following an analysis similar to the proof of Lemma~\ref{lemma:m-n+1}).}
	Hence it suffices to consider the case when every agent $i\in N$ has
	\begin{equation*}
	c_i(M_i^-) > \alpha\cdot c_i(\sigma_i(2)).
	\end{equation*}
	
	In other words, item $\sigma_i(2)$ (and thus every item other than $\sigma_i(1)$) is ``small''
	to agent $i$ in the sense that it contributes at most a $\frac{1}{\alpha+1}$ fraction to the total cost of $M - \sigma_i(1)$.
	Depending on whether $\sigma_i(1)$ has large cost, we classify the agents into two types: large agents and small agents.
    Here we introduce $\beta \in (0, 1/8)$ as the threshold to define large and small agents and our approximation ratio $\alpha$ will be decided by an optimized $\beta$.

	\begin{definition}[Large/Small Agent]
		We call agent $i\in N$ a \emph{large} agent if $c_i(\sigma_i(1)) \geq \beta$; \emph{small} otherwise.
	\end{definition}

	The main intuition behind the definition is as follows.
	For a large agent $i$, if $\sigma_i(1) \in X_j$ for some agent $j\neq i$, then as long as $c_i(X_i) \leq \alpha \cdot \beta$, agent $i$ is $\alpha$-EFX towards agent $j$, even if $X_j$ contains only one item $\sigma_i(1)$.
	On the other hand, for a small agent $i$, every item $e\in M$ has cost $c_i(e) < \beta$.
	Thus we can partition the items into bundles of roughly the same cost, under $c_i$, so that no matter which of these bundles agent $i$ eventually receives, she will be $\alpha$-EFX towards any other agent.
	Following the above intuitions, we proceed by considering how many small agents there are.
	In the following sections, we show that under different cases, the approximation ratio is either bounded by $\frac{2+4\beta}{1-4\beta}$ or $\frac{1}{2\beta}$.
    Therefore, by defining $\alpha = \max\left\{ \frac{2+4\beta}{1-4\beta},\frac{1}{2\beta} \right\}$, we can show that the algorithm is $\alpha$-approximate.
%
	Observe that when $0 < \beta < \frac{1}{8}$, $\frac{2+4\beta}{1-4\beta}$ is monotonically increasing and $\frac{1}{2\beta}$ is monotonically decreasing.
    Thus we can minimize $\alpha$ by picking $\beta$ as the positive root $\frac{\sqrt{6}-2}{4}$ of equation $\frac{2+4x}{1-4x} = \frac{1}{2x}$, which gives $\alpha = \max \left\{\frac{2+4\beta}{1-4\beta},\frac{1}{2\beta} \right\} =  2 + \sqrt{6} \approx 4.4495$.

	\subsection{At Least Two Small Agents}
	
	We first consider the case when there are at least two small agents.
	Without loss of generality, suppose agents $1$ and $2$ are small.
	Agent $3$ can be either small or large.
	
	\begin{lemma}\label{lemma:two-small}
		We can compute in polynomial-time a $3$-partition $(S_1, S_2, S_3)$ of $M$ such that for both $i\in \{1,2\}$, we have $c_i(S_j) \in [\frac{1}{4}-\beta, \frac{1}{2}+\beta]$ for all $j\in \{1,2,3\}$.
	\end{lemma}
	
	Note that Lemma~\ref{lemma:two-small} immediately implies Theorem~\ref{thm:n=3} when there are at least two small agents for the following reasons.
	Since the costs of the three bundles $S_1,S_2$ and $S_3$ differ by a factor of at most
    \begin{equation*}
        \frac{1/2+\beta}{1/4-\beta} = \frac{2+4\beta}{1-4\beta} \leq \alpha,
    \end{equation*}
    agent $1$ and $2$ are $\alpha$-EFX towards any other agent as long as every agent gets exactly one bundle.
	Therefore, by letting agent $3$ pick her favorite bundle, i.e., the one with minimum $c_3(S_j)$, and allocating the remaining two bundles to agents $1$ and $2$ arbitrarily, we end up with an $\alpha$-EFX allocation: agent $3$ does not envy agents $1$ and $2$; agents $1$ and $2$ are $\alpha$-EFX towards any other agent.
	Thus it remains to give the polynomial-time algorithm for the computation of $(S_1,S_2,S_3)$.
	The main idea behind the algorithm is quite simple: since agents $1$ and $2$ have small costs on every item, round-robin-like algorithms should work in computing such a partition.
	
	\paragraph{The Algorithm}
	We initialize $S_j$ as an empty bundle for all $j\in \{1,2,3\}$.
	Then we let agents $1$ and $2$ take turns to put the unallocated item with maximum cost into the bundle with the smallest cost, both under their own cost function until all items are allocated (see Algorithm~\ref{alg:RRP}).
	
	\begin{algorithm}
		\caption{Sequential Placement}\label{alg:RRP} 
		Initialize: $S_j \gets \emptyset$ for all $j\in \{1,2,3\}$, $P\gets M$ and $k\gets1$  \;
		\While{$P\neq\emptyset$}{
			let $e^*\leftarrow\argmax\{{c_k(e):e\in P}\}$ and $j^*\leftarrow\argmin\{{c_k(S_j):j\in \{1,2,3\}}\}$\;
			$S_{j^*}\leftarrow S_{j^*} + e^*,P\leftarrow P-e^*$\;
			$k\leftarrow (k\bmod 2)+1$ \;
		}
		\KwOut{$(S_1,S_2,S_3)$}
	\end{algorithm}
	
	Next we show that for both $i\in\{1,2\}$, we have $c_i(S_j) \in [1/4-\beta, 1/2+\beta]$ for all $j\in \{1,2,3\}$.
	
	\begin{proofof}{Lemma~\ref{lemma:two-small}}
		Fix any agent $i\in \{1,2\}$, say $i=1$.
		We assume w.l.o.g. that $c_1(S_1)\leq c_1(S_2)\leq c_1(S_3)$.
		We show that $c_1(S_3) \leq 1/2+\beta$ and $c_1(S_1) \geq 1/4-\beta$.
		Recall that each agent allocates one item to some bundle in each round.
		We say that it is agent $i$'s round if during this round agent $i$ gets to allocate an item to one of the bundles.	

		Consider the last round during which agent $1$ allocates an item to the bundle $S_3$, and let $t$ be the \emph{next} agent $1$'s round.
        {  If $t$ does not exist, e.g., agent $1$ allocates an item to bundle $S_3$ in her last round, then we let $t$ be a dummy round that happens after the sequential placement and in round $t$ agent $1$ allocates a dummy item (with cost $0$) to some bundle other.}
		If agent $1$ never assigns any item to $S_3$, then let $t$ be the first agent $1$'s round.
		In other words, $t$ is defined as the earliest agent $1$'s round such that starting from round $t$, agent $1$ never assigns any item to bundle $S_3$.
		Let $S'_1, S'_2, S'_3$ be the three bundles at the beginning of round $t$.
		Let $P' = M\setminus (S'_1\cup S'_2\cup S'_3)$ be the set of unallocated items at this moment.
        \begin{itemize}
        \item If $t$ is the first agent $1$'s round, then we have $c_1(S'_3) \leq \beta$, since at most one item is allocated to $S'_3$ and each item has cost at most $\beta$ under $c_1$.
		
		\item Otherwise, since agent $1$ allocates one item to $S_3$ in the previous agent $1$'s round, we know that at the beginning of that round, the cost of bundle $S_3$ is the minimum among the three.
		Thus we have
		\begin{equation}\label{eq:dif-of-S1-S3}
		c_1(S'_3) \leq c_1(S'_1) + 2\beta,
		\end{equation}
		since at most one item is allocated in each round. 
        \end{itemize}
        
		In both cases we have $c_1(S'_3) \leq  c_1(S'_1) + 2\beta$.
		
		\smallskip
		
		Let $P_1$ (resp. $P_2$) be the set of items in $P'$ that are allocated by agent $1$ (resp. $2$).
		Note that all these items are allocated at or after round $t$.
		Since agent $1$ allocates the items from the most costly to the least, we have $c_1(P_1) \geq c_1(P_2)$.
		We first give an upper bound on $c_1(S_3)$.
		
		\begin{claim}\label{claim:S3+P2}
			We have $c_1(S'_3) + c_1(P_2) \leq 1/2+\beta$.
		\end{claim}
		\begin{proof}
			Recall that
			\begin{equation*}
			    c_1(P_2) \leq \frac{1}{2}\cdot c_1(P') = \frac{1}{2}\cdot (1-c_1(S'_1)-c_1(S'_2)-c_1(S'_3)).
			\end{equation*}
			We have
			\begin{align*}
			c_1(S'_3) + c_1(P_2) \leq\ & c_1(S'_3) + \frac{1}{2}\cdot (1-c_1(S'_1)-c_1(S'_2)-c_1(S'_3)) \\
			\leq\ & \frac{1}{2} + \frac{1}{2}\cdot (c_1(S'_3)-c_1(S'_1))
			\leq \frac{1}{2} + \beta,
			\end{align*}			
			where the last inequality follows from~\eqref{eq:dif-of-S1-S3}.
		\end{proof}
		
		Recall that all items in $S_3 \setminus S'_3$ are allocated by agent $2$ after round $t$, by Claim~\ref{claim:S3+P2} we have
		\begin{equation*}
		c_1(S_3) \leq c_1(S'_3) + c_1(P_2) \leq \frac{1}{2}+\beta.
		\end{equation*}
		
		Next we give a lower bound for $c_1(S_1)$.
		If agent $1$ allocates all items in $P_1$ to $S_1$, then we have
		\begin{align*}
		c_1(S_1) \geq c_1(S'_1) + c_1(P_1) \geq c_1(S'_3) - 2\beta + c_1(P_2) \geq c_1(S_3) - 2\beta.
		\end{align*}
		
		Consequently we have (recall that $c_1(S_2) \leq c_1(S_3)$)
		\begin{equation*}
		1 = c_1(S_1) + c_1(S_2) + c_1(S_3) \leq 3\cdot c_1(S_1) + 4\beta,
		\end{equation*}
		which gives $c_1(S_1) \geq \frac{1-4\beta}{3} > \frac{1}{4}-\beta$.

	    \begin{figure}[htb]
            \centering
            \includegraphics[width=0.3\textwidth]{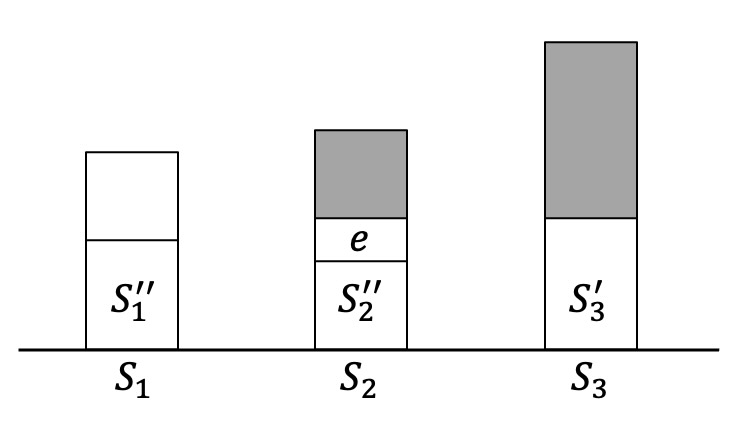}
            \caption{An illustration of the assignment of items in the three bundles, where the items within each bundle are ordered (from bottom to top) by the time they are included (from earliest to latest), and the items in the shaded areas belong to $P_2$.}
            \label{fig:P2_analysis}
        \end{figure}

		Otherwise we consider the last item $e\in P_1$ agent $1$ allocates to $S_2$ (see Figure~\ref{fig:P2_analysis} for an illustrative example).
		Let $S''_1$ and $S''_2$ be the two bundles right before the allocation.
		Note that $S'_1 \subseteq S''_1$, $S'_2 \subseteq S''_2$ and we have $c_1(S''_1) \geq c_1(S''_2)$.
		Since $e$ is the last item agent $1$ allocates to $S_2$, we know that all items in $S_2 \setminus (S''_2 + e)$ are from $P_2$.
		Also recall that after time $t$, agent 1 did not assign any item to $S_3$.
		In other words, all items in $S_3\setminus S_3'$ are from $P_2$.
		Thus we have
		\begin{equation*}
		    M\setminus (S_1 \cup S''_2 \cup \{e\}) \subseteq S_3 \cup P_2 = S'_3 \cup P_2.
		\end{equation*}
		By Claim~\ref{claim:S3+P2} we have
		\begin{equation*}
		1 - c_1(S_1) - c_1(S''_2) - c_1(e) \leq c_1(S'_3) + c_1(P_2) \leq \frac{1}{2}+\beta.
		\end{equation*}
		
		Since $c_1(S_1) \geq c_1(S''_1) \geq c_1(S''_2)$, we have
		\begin{equation*}
		c_1(S_1) \geq \frac{1}{2}\cdot \left(1-\frac{1}{2} -\beta - c_1(e)\right) \geq \frac{1}{4}-\beta,
		\end{equation*}
		where the second inequality holds due to $c_1(e) \leq \beta$.
	\end{proofof}

	\subsection{One Small Agent}
	
	Next we consider the case when there is exactly one small agent, say agent $3$.
	Let $e_1 = \sigma_1(1)$ and $e_2 = \sigma_2(1)$ be the most costly item under $c_1$ and $c_2$, respectively.
	Note that since agents $1$ and $2$ are large, we have $c_1(e_1) \geq \beta$ and $c_2(e_2) \geq \beta$.
	It is possible that $e_1=e_2$.

	\begin{lemma}\label{lemma:one-small}
	    When there is only one small agent, an $\alpha$-EFX allocation can be computed in polynomial time.
	\end{lemma}
	
	If $e_1=e_2$, we can assign it to agent 3.
    Let $X_3 = \{ e_1 \}$, and we compute an EFX allocation $(X_1, X_2)$ between agents $1$ and $2$ on items $M - e_1$.
    In this allocation, agent $3$ is obviously $\alpha$-EFX towards agents $1$ and $2$ because $|X_3| = 1$, and agents $1$ and $2$ do not envy each other by more than one item.
	On the other hand, since $(X_1, X_2)$ is an EFX allocation on items $M - e_1$, by removing any item $e$ from $X_1$ (resp. $X_2$), we have $c_1(X_1 - e) \leq 1/2$ (resp. $c_2(X_2 - e) \leq 1/2$).
	Since $c_1(X_3) = c_1(e_1) \geq \beta$ and $c_2(X_3) = c_2(e_2) \geq \beta$, the allocation is $\alpha$-EFX {  (recall that $\alpha \geq 1/(2\beta)$)}.
	
    Next, we consider the other case that $e_1\neq e_2$ and use the fact that agent $3$ is a small agent to compute an allocation that is fair to all.

    \paragraph{The Algorithm}
	If $e_1\neq e_2$, let $M^- = M - e_1 - e_2$.
    We first compute an EFX allocation $(S_1,S_2,S_3)$ for three agents under cost function $c_3$, on items $M^-$.
	Assume w.l.o.g. that $c_1(S_1)\leq c_1(S_2)\leq c_1(S_3)$, we decide the allocation $\bX$ as follows.
	Let $X_1 = S_1 + e_2$, let agent $2$ pick a bundle between $S_2 + e_1$ and $S_3$; then agent $3$ gets the remaining bundle.
	
    \begin{algorithm}
		\caption{Algorithm for 2 Large Agents}\label{alg:PDC2} 
		Initialize: $X_j \gets \emptyset$ for all $j\in \{1,2,3\}$, $P\gets M^-$  \;
		Compute an EFX allocation $(S_1, S_2, S_3)$ on items $P$ under cost function $c_3$ \;
		$X_1 \gets S_1+e_2$, $X_2 \gets \argmin\{c_2(S_2+e_1), c_2(S_3)\}$ \;
		$X_3\gets M\setminus (X_1\cup X_2)$ \;
		\KwOut{$(X_1,X_2,X_3)$}
	\end{algorithm}
	
	\begin{proofof}{Lemma~\ref{lemma:one-small}}
	    We show that the allocation is $\alpha$-EFX.
	    Let $c_3(S_i) = \min\{c_3(S_1),c_3(S_2),c_3(S_3)\}$, we show that $c_3(S_i)\geq \frac{1-4\beta}{3}$.
        Note that we have $c_3(M^-) \geq 1-2\beta$ since items $e_1$ and $e_2$ have cost at most $\beta$ under $c_3$.
	    Hence
        \begin{equation*}
	        c_3(M^-) = c_3(S_1) + c_3(S_2) + c_3(S_3) \leq 3\cdot c_3(S_i) + 2\beta, 
	    \end{equation*}
	    which implies $c_3(S_i) \geq \frac{1}{3}\cdot (c_3(M^-) - 2\beta) \geq \frac{1-4\beta}{3}$.

        In the following, we show that each agent $i\in \{1,2,3\}$ is $\alpha$-EFX towards the other two agents.

        {  Recall that $c_3(M) = 1$ and $\min\{ c_3(X_1),c_3(X_2) \} \geq \frac{1-4\beta}{3}$.
        Agent $3$ is $\alpha$-EFX towards agents $1$ and $2$ since
        }
		\begin{equation*}
		    c_3(X_3) \leq \frac{1+8\beta}{3} \leq \frac{1+8\beta}{1-4\beta}\cdot \min\{ c_3(X_1),c_3(X_2) \} \leq (\alpha-1) \cdot \min\{ c_3(X_1),c_3(X_2) \}.
		\end{equation*}
		The last inequality holds because $\frac{1+8\beta}{1-4\beta} = \frac{2+4\beta}{1-4\beta}-1 \leq \alpha-1$.
		
		Agent $2$ does not envy agent $3$ because she gets to pick a bundle between $S_2 + e_1$ and $S_3$, while agent $3$ gets the other bundle.
		For the same reason, we also have $c_2(X_2) \leq 1/2$. 
		Agent $2$ is $\alpha$-EFX towards agent $1$ because $c_2(X_1) \geq c_2(e_2) \geq \beta$, $c_2(X_2) \leq \frac{1}{2\beta}\cdot c_2(X_1) \leq \alpha\cdot c_2(X_1)$.
		
		Agent $1$ does not envy bundle $S_2 + e_1$ because $c_1(S_1) \leq c_1(S_2)$ and $c_1(e_2) \leq c_1(e_1)$.
		Agent $1$ is $\alpha$-EFX towards bundle $S_3$ because $c_1(S_1) \leq c_1(S_3)$ and
		\begin{equation*}
		c_1(e_2) \leq \frac{1}{\alpha}\cdot c_1(M_1^-) \leq \frac{1}{\alpha}\cdot c_1(M^-) \leq \frac{3}{\alpha}\cdot c_1(S_3),
		\end{equation*}
		where the last inequality holds since $S_3$ is the most costly bundle among $\{ S_1,S_2,S_3 \}$, i.e., $c_1(S_3)\geq \frac{1}{3}\cdot c_1(M^-)$.
		Hence we have $c_1(S_1+e_2) \leq \frac{3+\alpha}{\alpha}\cdot c_1(S_3)$.
        {  Note that $\frac{3+\alpha}{\alpha} \leq \alpha$ for all $\alpha \geq \frac{1+\sqrt{13}}{2}\approx 2.3028$.}
		Thus no matter which bundle agent $2$ picks, the resulting allocation is $\alpha$-EFX.
	\end{proofof}

	\subsection{All Agents are Large}

	Finally, we consider the case when all agents are large.
	Let $e_1 = \sigma_1(1)$, $e_2 = \sigma_2(1)$ and $e_3 = \sigma_3(1)$.
	By definition we have $c_1(e_1)\geq \beta$, $c_2(e_2)\geq \beta$ and $c_3(e_3)\geq \beta$.
    It is possible that some of $e_1,e_2,e_3$ are referring to the same item.
 
	\begin{lemma}\label{lemma:all-large}
	    When all three agents are large, an $\alpha$-EFX allocation can be computed in polynomial time.
	\end{lemma}
	
	As before, if there exist two of $\{ e_1,e_2,e_3 \}$ that are the same item, say $e_1 = e_2$, then we can easily compute an $(1/2\beta)$-EFX allocation by assigning $e_1$ to $X_3$ and computing an EFX allocation between agent $1$ and $2$ on the remaining items.
	Hence we assume $e_1, e_2$ and $e_3$ are three different items, and let $M^- = M\setminus \{ e_1,e_2,e_3 \}$.
	
	\paragraph{The Algorithm}
	We initialize $X_j$ as an empty bundle for all $j\in \{1,2,3\}$.
	We first assign both $e_1$ and $e_2$ to agent 3, and {  let} agent 3 quit the allocation.
	Then we compute an EFX allocation $(S_1, S_2)$ on items $M^-$ under cost function $c_1$.
	Assume w.l.o.g. that $c_3(S_1) \leq c_3(S_2)$.
	We let agent $2$ pick a bundle between $S_1 + e_3$ and $S_2$, and assign the other bundle to agent $1$ (refer to Algorithm~\ref{alg:PDC}).
	
	\begin{algorithm}
		\caption{Algorithm for 3 Large Agents}\label{alg:PDC} 
		Initialize: $X_j \gets \emptyset$ for all $j\in \{1,2,3\}$, $P\gets M^-$  \;
		let $X_3\gets \{e_1,e_2\}$ \;
		Compute an EFX allocation $(S_1, S_2)$ on items $P$ under cost function $c_1$ \;
		let $S_{i^*} \gets \argmin\{c_3(S_1), c_3(S_2)\}$ \;
		{  $X_2\gets \argmin\{c_2(S_{i^*}+e_3), c_2(P \setminus S_{i^*})\}$ \;}
		$X_1\gets (P\setminus X_2) +e_3$ \;
		\KwOut{$(X_1,X_2,X_3)$}
	\end{algorithm}
	
	\begin{proofof}{Lemma~\ref{lemma:all-large}}
	We show that when all three agents are large, and $e_1, e_2$ and $e_3$ are different items, Algorithm~\ref{alg:PDC} returns an $\alpha$-EFX allocation.
	\begin{claim}\label{claim:M^-&M_i^-}
	    For all $i\in \{1,2,3\}$, we have $c_i(M^-)\geq \frac{\alpha-1}{\alpha}\cdot c_i(M_i^-)$.
	\end{claim}
	\begin{proof}
	    Recall that $M_i^- = M\setminus \{\sigma_i(1),\sigma_i(2)\}$ and for every agent $i\in N$, we have $c_i(M_i^-)> \alpha\cdot c_i(\sigma_i(2))$.
	    Hence for all $i\in \{1,2,3\}$, we have
	    \begin{align*}
	        c_i(M^-) &= c_i(M)- c_i(e_1+e_2+e_3) 
	        \geq c_i(M) - c_i(\sigma_i(1)) - 2\cdot c_i(\sigma_i(2)) \\
	        &\geq c_i(M_i^-) - c_i(\sigma_i(2)) 
	        \geq c_i(M_i^-) - \frac{1}{\alpha}\cdot c_i(M_i^-) 
	        \geq \frac{\alpha-1}{\alpha}\cdot c_i(M_i^-).
	    \end{align*}
	    Thus Claim~\ref{claim:M^-&M_i^-} holds.
	\end{proof}

    In the following, we show that each agent $i\in \{1,2,3\}$ is $\alpha$-EFX towards the other two agents.

	\begin{itemize}
		\item Agent $3$ is $\alpha$-EFX towards bundle $S_1 + e_3$ because $c_3(X_3) \leq 2\cdot c_3(e_3)$; is $\alpha$-EFX towards bundle $S_2$ because
		\begin{align*}
		c_3(S_2) & \geq \frac{1}{2}\cdot c_3(M^-) \geq \frac{\alpha-1}{2\alpha}\cdot c_3(M_3^-) > \frac{\alpha-1}{4}\cdot c_3(X_3),
		\end{align*}
		where the second inequality hold because of Claim~\ref{claim:M^-&M_i^-}, the third inequality holds because $c_3(M_3^-) > \alpha\cdot c_3(e_1)$ and $c_3(M_3^-) > \alpha\cdot c_3(e_2)$.
		Since $\frac{4}{\alpha-1} \leq \alpha$ for all $\alpha \geq \frac{1+\sqrt{17}}{2}\approx 2.5616$, we have $c_3(X_3) < \alpha\cdot c_3(S_2)$.
		
		\item Agent $2$ does not envy agent $1$ because she gets to pick first; is $\alpha$-EFX towards agent $3$ because $c_2(X_2) \leq 1/2$ while $c_2(X_3) \geq c_2(e_2) \geq \beta$.
		
		\item Recall that agent $1$ gets the bundle in $\{ S_1 + e_3, S_2\}$ that is not picked by $2$.
		Moreover, $S_1\cup S_2 = M^-$ and $(S_1,S_2)$ is an EFX allocation under $c_1$.
		Since every item $e \neq e_1$ has cost $c_1(e) \leq \frac{1}{\alpha}\cdot c_1(M_1^-)$, we have $c_1(e) \leq \frac{1}{\alpha-1}\cdot c_1(M^-)$ for every $e\in M^-$.
		Consequently, we have
		\begin{align*}
		 \max \{ c_1(S_1), c_1(S_2) \} \leq \frac{\alpha}{2\alpha-2}\cdot c_1(M^-)\quad \text{ and } \quad
		 \min \{ c_1(S_1), c_1(S_2) \} \geq \frac{\alpha-2}{2\alpha-2}\cdot c_1(M^-),
		\end{align*}
		because otherwise {  $|c_1(S_1) - c_1(S_2)| > \frac{1}{\alpha-1}\cdot c_1(M_1^-)$}, which contradicts the fact that $(S_1,S_2)$ is an EFX partition.
		Hence we have
		\begin{align*}
		\max \{ c_1(S_1 + e_3), c_1(S_2) \}
		\leq\ & \frac{\alpha}{2\alpha-2} \cdot c_1(M^-) + \frac{1}{\alpha-1} \cdot c_1(M^-) \\
		= \ & \frac{\alpha+2}{2\alpha-2} \cdot c_1(M^-)
		\leq  \frac{\alpha+2}{\alpha-2} \cdot \min \{ c_1(S_1 + e_3), c_1(S_2) \} \\
		\leq\ & \alpha \cdot \min\{c_1(S_1+e_3), c_1(S_2)\}.
		\end{align*}
        The last equality holds since $\frac{\alpha+2}{\alpha-2} \leq \alpha$ holds for all $\alpha \geq \frac{3+\sqrt{17}}{2}\approx 3.5616$.
		Hence agent $1$ is $\alpha$-EFX towards agent $2$.
		In the following, we show that for all $e\in X_1$ we have $c_1(X_1 - e) \leq 1/2$. 
		If $X_1 = S_2$ then the claim is trivially true because $(S_1, S_2)$ is an EFX partition of $M^-$.
		Otherwise we have $X_1 = S_1 + e_3$.
		If $e = e_3$, then
		\begin{equation*}
		    c_1(X_1 - e) = c_1(S_1)  \leq c_1(S_2 + e_1).
		\end{equation*}
		If $e\neq e_3$, then we have 
		\begin{equation*}
		    c_1(X_1 - e) = c_1(S_1 - e) + c_1(e_3) \leq c_1(S_2 + e_1).
		\end{equation*}
		
		In both cases we have $c_1(X_1 - e) \leq c_1(S_2 + e_1)$.
		Since $X_1 = S_1 + e_3$ and $S_2 + e_1$ are disjoint, we have
		\begin{equation*}
		    c_1(X_1 - e) \leq \frac{1}{2}\cdot (c_1(X_1 - e) + c_1(S_2 +e_1)) \leq \frac{1}{2}.
		\end{equation*}
		Since $c_1(X_3) \geq c_1(e_1) \geq \beta$, we have $c_1(X_1-e) \leq \frac{1}{2\beta}\cdot c_1(X_3) $.
		Hence agent $1$ is $\alpha$-EFX towards agent $3$ either.
	\end{itemize}

    Hence any agent is $\alpha$-EFX towards any other agent, and the allocation is $\alpha$-EFX.
	\end{proofof}
	
	We argue that Algorithms~\ref{alg:RRP},~\ref{alg:PDC2} and~\ref{alg:PDC} run in $O(m \log m)$ time because under a fixed cost function, computing an EFX allocation can be done by sorting items in descending order of costs and allocating items sequentially to form a partition.
	Given this partition, the final allocation can be determined in $O(1)$ time.
	In summary, in all cases, we can compute a $(2+\sqrt{6})$-EFX allocation for three agents in polynomial-time, which proves Theorem~\ref{thm:n=3}.
 
	From the above analysis, we observe that it is crucial to distinguish whether an item $e$ is large to an agent $i$: if it is, then by allocating the item to another agent $j$ who values it small, we can ensure that $i$ is $\alpha$-EFX towards $j$; if it is not, then putting item $e$ in $X_i$ does not hurt the approximation ratio too much.
	In the following section, we show how these ideas can be extended to the general case when $n\geq 4$.

	\section{Four or More Agents}\label{sec:fourormore}
	
	In this section, we give a polynomial-time algorithm that computes a $(3n^2-n)$-EFX allocation for any given instance with $n\geq 4$ agents.
	
	\begin{theorem}\label{thm:general-n}
		There exists an algorithm that computes a $(3n^2-n)$-EFX allocation for any instance with $n$ agents in $O(nm \log m)$ time.
	\end{theorem}
	
	Recall that for each agent $i\in N$, we define $M_i^- = \{ \sigma_i(j) : j\geq n \}$ as the set of tail items.
	As before (refer to Lemma~\ref{lemma:m-n+1} for a formal analysis), if there exists an agent $i$ with
	\begin{equation}
	    c_i(M_i^-) \leq (3n^2-n)\cdot c_i(\sigma_i(n-1)), \label{eq:simple_case_n>=4}
	\end{equation}
	then we can easily compute a $(3n^2-n)$-EFX allocation by allocating exactly one item in $M\setminus M_i^-$ to each agent in $N\setminus\{i\}$, and assigning the remaining items $M_i^-$ to agent $i$.
	Since each agent other than $i$ receives only one item, and agent $i$ receives a set of items with a total cost at most $3n^2-n$ times the cost of any other agent, the allocation is $(3n^2-n)$-EFX.
    From now on we assume that {  for every agent $i$,} Equation~\eqref{eq:simple_case_n>=4} is not true.
    

	\subsection{The Allocation Algorithm}
	
	We define $b_i = \frac{1}{3n^2 - 2n + 2}\cdot c_i(M_i^-)$, and let $L_i = \{ e\in M : c_i(e) \geq b_i \}$ be the set of \emph{large} items of agent $i$.
	Note that by the above discussion we have $|L_i| \leq n-2$ for all $i\in N$ because otherwise
    \begin{equation*}
        c_i(M_i^-) = (3n^2 - 2n + 2)\cdot b_i \leq (3n^2-n)\cdot c_i(\sigma_i(n-1)),
    \end{equation*} 
    which {  satisfies} Equation~\eqref{eq:simple_case_n>=4}.
	The main intuition behind the definition of large items is as follows.
	Our algorithm will compute an allocation $\bX = (X_1,\ldots,X_n)$ ensuring that for each agent $i\in N$, either $X_i \cap L_i = \emptyset$, i.e., no large item in $L_i$ is assigned to agent $i$; or $|X_i| = 1$.
	We show that as long as no large item is assigned to agent $i$, $i$ is $(3n^2-n)$-EFX towards any other agent that receives at least one item in $L_i$.
	
	\begin{lemma}\label{lemma:only-small-items}
		For any agent $i\in N$, if $X_i\cap L_i = \emptyset$ and $L_i \cap X_j \neq \emptyset$, then $i$ is $(3n^2-n)$-EFX towards agent $j$.
	\end{lemma}
	\begin{proof}
		By definition each item $e\in L_i$ has cost 
		\begin{equation*}
		    c_i(e) \geq b_i = \frac{1}{3n^2 - 2n +2}\cdot c_i(M_i^-).
		\end{equation*}
        {  Note that $M\setminus M_i^- = \{\sigma_i (j): j\leq n-1\}$.
        Also note that $L_i \subseteq M\setminus M_i^-$ and $|L_i| \geq 1$ since $L_i \cap X_j \neq \emptyset$.
        In other words, there are at most $n-2$ items that are not included in $M_i^-\cup L_i$, all of which are small to agent $i$.
		Hence we have
		\begin{align*}
		c_i(X_i) & \leq c_i(M\setminus L_i) = c_i(M_i^-) + c_i(M\setminus (M_i^- \cup L_i)) \\
        & \leq (3n^2 - 2n + 2)\cdot b_i + (n-2)\cdot b_i \\
		& = (3n^2-n) \cdot b_i \leq (3n^2-n) \cdot c_i(X_j),
		\end{align*}
        where the last inequality holds because $X_j$ contains at least one item that is large to agent $i$.
        }
	\end{proof}
	
	\paragraph{Remark}
	A few difficulties arise when we try to extend the ideas we develop for three agents to general number of agents.
	First, for a large number of agents, it is no longer feasible to classify agents depending on how many large items they have because there are too many cases.
	Instead, our new algorithm removes the large items of each agent, and treats all agents as ``small agents''.
	Second, even if all agents value all items small, i.e., $L_i = \emptyset$ for all $i\in N$, it is not clear how to extend the Sequential Placement algorithm (Algorithm~\ref{alg:RRP}) to compute $O(n)$ subsets with upper and lower bounded costs for general number of agents.
	To get around this, we borrow existing results for the PROP1 allocation of goods, and show that when items are small to the agents, we can partition the items into $O(n)$ bundles such that the ratio between the costs of any two bundles is bounded by $3n^2-n$.
	
	\smallskip
	
	Let $L = \cup_{i\in N} L_i$ be the set of items that are large to at least one agent.
	Let $K = \cap_{i\in N} L_i$ be the set of items that are large to all agents.
	Let $M^- = M\setminus L$ be the set of items that are small to all agents.
	
	\begin{claim}\label{claim:cost-of-small}
		For each $i\in N$ and $e\in M^-$, we have $c_i(e) \leq \frac{1}{2n^2 + n}\cdot c_i(M^-)$.
	\end{claim}
	\begin{proof}
		Recall that item $e\in M^-$ is small to agent $i$.
		Thus
		\begin{equation*}
		    c_i(e) \leq b_i = \frac{1}{3n^2 - 2n + 2}\cdot c_i(M_i^-).
		\end{equation*}
		Moreover, we have $M_i^- \setminus (\cup_{j\in N} L_j \setminus L_i) \subseteq M^-$.
		Since there are at most $(n-1)\cdot (n-2)$ items in $\cup_{j\in N} L_j \setminus L_i$, and each of them has cost at most $b_i$ to agent $i$, we have
		\begin{equation}
		c_i(M^-) \geq c_i(M_i^-) - (n-1)\cdot (n-2)\cdot b_i \geq (2n^2 + n)\cdot b_i.\label{eq:M_i^-and-M^-}
		\end{equation}
		
		By the above lower bound on $c_i(M^-)$, we have $c_i(e) \leq b_i \leq \frac{1}{2n^2 + n}\cdot c_i(M^-)$.
	\end{proof}
	
	Recall that now we have three sets of items $K$, $M^-$ and $L\setminus K$, each of which will be handled as follows.
	\begin{itemize}
		\item Since each item in $K$ is large to all agents, we assign each of them to a unique agent chosen arbitrarily.
		Let $N^*$ be these agents, and $N^- = N\setminus N^*$.
		Note that $|N^*| = |K|$.
		Our algorithm will not assign any further items to agents in $N^*$.
		Obviously, these agents are $(3n^2-n)$-EFX towards any other agents.
		Moreover, if we can ensure that in the final allocation $X_i\cap L_i = \emptyset$ for all $i\in N^-$, then all agents are $(3n^2-n)$-EFX towards agents in $N^*$, by Lemma~\ref{lemma:only-small-items}.
		
		\item By Claim~\ref{claim:cost-of-small}, for every agent $i$, all items in $M^-$ have small cost compared to $c_i(M^-)$.
		We show in Lemma~\ref{lemma:even-partition-general-n} that we can partition the items in $M^-$ into $|N^-|$ bundles {  such that each agent in $N^-$ perceives all other bundles with some guaranteed lower bound on the cost}.
		The key to the computation of the partition is to ensure that each bundle has a considerably large cost to every agent in $N^-$.
		
		\item It remains to assign items in $L\setminus K$.
		Recall that these are items that are large to some agents but not to all agents.
		Our algorithm assigns each $e \in L\setminus K$ to an arbitrary agent $i$ for which $e\notin L_i$.
	\end{itemize}
	
	\begin{lemma}\label{lemma:even-partition-general-n}
		Given a set of items $M^-$ and a group of agents $N^-$ such that for all $i\in N^-$ and $e\in M^-$, $c_i(e)\leq \frac{1}{2n^2 + n}\cdot c_i(M^-)$,
		there exists a partition $(S_1,S_2,\ldots,S_{|N^-|})$ of $M^-$ such that for all $i,j\in N^-$,
		\begin{equation*}
		    c_i(S_j) \geq \frac{1}{2n^2 + n}\cdot c_i(M^-).
		\end{equation*}
	\end{lemma}
	
    \begin{proof}
	Let $N^- = \{1,2,\ldots,t\}$.
	Our goal is to partition $M^-$ into $(S_1,S_2,\ldots,S_{t})$ such that for all $i,j\in N^-$, $c_i(S_j) \geq \frac{1}{2n^2 + n}\cdot c_i(M^-)$.
	Since we need to partition the items into bundles whose cost is lower bounded, it is natural to borrow ideas from the allocation of goods.
	In particular, we treat items in $M^-$ as goods, where each agent $i$ has \emph{value} $c_i(e)$ on item $e\in M^-$. Then we compute a $t$-partition $\mathbf{Q} = (Q_1,Q_2,\ldots,Q_t)$ of $M^-$ that is PROP1 using the Round Robin algorithm~\cite{conf/sigecom/LiptonMMS04}.
	The fairness notion PROP1 is defined as follows:
	An allocation $\mathbf{Q}$ of items $M^-$ to agents $N^-$ is PROP1 if for all $i\in N^-$, there exists $e\in M^-\setminus Q_i$ such that $c_i(Q_i + e) \geq \frac{1}{t}\cdot c_i(M^-)$.
	
	Since each $e\in M^-$ has cost $c_i(e) \leq \frac{1}{2n^2 + n}\cdot c_i(M^-)$ to agent $i$, we have
	\begin{align}
	c_i(Q_i) &\geq \frac{1}{n}\cdot c_i(M^-) - \frac{1}{2n^2 + n}\cdot c_i(M^-)
	\geq \frac{2n}{2n^2 + n}\cdot c_i(M^-). \label{eq:lower-bound-Qi}
	\end{align}
	
	However, $\mathbf{Q}$ is not the desired partition since $c_i(Q_j)$ can be arbitrarily small, for some $j\neq i$.
	Thus for each $i\in N^-$ we further compute a $t$-partition $(Q_{i1},Q_{i2},\ldots,Q_{it})$ of $Q_i$ that is PROP1 under $c_i$. 
	Again, since each $e\in Q_i$ has $c_i(e)\leq \frac{1}{2n^2 + n}\cdot c_i(M^-)$, for each $j\leq t$ we have
	\begin{align}
	c_i(Q_{ij}) & \geq \frac{1}{n}\cdot c_i(Q_i) - \frac{1}{2n^2 + n}\cdot c_i(M^-) \geq \frac{1}{2n^2 + n}\cdot c_i(M^-),
	\label{eq:lower-bound-Qij}
	\end{align}
	where the second inequality follows from Equation~\eqref{eq:lower-bound-Qi}.
	
	Finally, we regroup the bundles $\{ Q_{ij} \}_{i,j\in N^-}$ into $t$ bundles: let $S_j = \cup_{i=1}^t Q_{ij}$ for each $j\in N^-$.
	By Equation~\eqref{eq:lower-bound-Qij}, for any $i, j\in N^-$, we have
	\begin{equation*}
	c_i(S_j) \geq c_i(Q_{ij}) \geq \frac{1}{2n^2 + n}\cdot c_i(M^-),
	\end{equation*}
	which concludes the proof of Lemma~\ref{lemma:even-partition-general-n}.
    \end{proof}
	
	\subsection{The Approximation Ratio}
	
	In this section, we show that all agents are $(3n^2-n)$-EFX towards the other agents.
	From the previous analysis, we know that agents in $N^*$ do not envy any other agent by more than one item, and agents in $N^-$ are $(3n^2-n)$-EFX towards agents in $N^*$.
	It remains to show that agents in $N^-$ are $(3n^2-n)$-EFX towards each other.
	Recall that each agent $i\in N^-$ receives a bundle $S_i$ (see Lemma~\ref{lemma:even-partition-general-n}), and possibly some other items from $L\setminus K$ that are small to agent $i$.
	By Lemma~\ref{lemma:even-partition-general-n}, for every $j\in N^-$ we have $c_i(X_j) \geq \frac{1}{2n^2 + n}\cdot c_i(M^-)$.
	
	Next we give an upper bound on $c_i(X_i)$.
	{  Recall that for any $i\in N$ we have $|L_i|\leq n-2$.
	Hence there are at most $(n-1)\cdot (n-2)$ items in $L\setminus K$ that are small to agent $i$.
    Furthermore, there are at most $n-2$ items in $K$, which leads to $|N^-| \geq 2$.
    In other words, for any agent $i\in N^-$, there exists at least another agent $j\in N^-$ such that $c_i(X_j) \geq \frac{1}{2n^2+n}\cdot c_i(M^-)$.}
        We have
	\begin{align*}
	c_i(X_i) & \leq (1-\frac{1}{2n^2 + n})\cdot c_i(M^-) + (n-1)\cdot (n-2)\cdot b_i \\
	& \leq \frac{2n^2+n-1}{2n^2+n}\cdot c_i(M^-) + \frac{n^2-3n+2}{2n^2+n}\cdot c_i(M^-)\\
	& \leq \frac{3n^2 - 2n + 1}{2n^2 + n}\cdot c_i(M^-) \leq \frac{3n^2-n}{2n^2 + n}\cdot c_i(M^-),
	\end{align*}
	where the second inequality follows from Equation~\eqref{eq:M_i^-and-M^-}.
	Combining the upper bound on $c_i(X_i)$ and lower bound on $c_i(X_j)$, we have $c_i(X_i) \leq (3n^2-n)\cdot c_i(X_j)$ for any $i,j\in N^-$.
	Hence agents in $N^-$ are $(3n^2-n)$-EFX towards each other.
	
    {  Finally, we argue that the allocation can be computed in $O(n m \log m)$ time.
    The main complexity comes from the division of items into three sets $K, M^-$ and $L\setminus K$, which takes $O(n m \log m)$ time as each $L_i$ can be computed in $O(m \log m)$ time by sorting items based on $c_i$.
    Given the three sets, it can be verified that allocating items in $K$, $M^-$ and $L\setminus K$ takes $O(n m \log m)$ time.
    More specifically, the computation of bundles $\{Q_{ij}\}_{i,j\in N^-}$ takes $O(nm \log m)$ time by using the Round-Robin algorithm twice.
    The regrouping of those bundles can be done in $O(nm)$ time since each bundle $S_j$ contains at most $m$ items and there are $|N^-|<n$ bundles to compute.}

	\section{Bi-valued Instances with Three Agents}\label{sec:bi-three}
	
	In the following two sections, we consider the fair allocation problem with agents having bi-valued cost functions\footnote{In these two sections we use some notations that are different from before, which are more convenient to work with for bi-valued instances. For example, we no longer assume the cost functions are normalized and we do not use $\sigma_i(j)$ to denote the $j$-th most costly item.}.
	That is, there exist constants $a, b\geq 0$ such that for any $i\in N$ and $e\in M$, $c_i(e)\in\{a,b\}$.
	Equivalently, for any $a\neq b$, we can scale the cost function so that $c_i(e)\in\{\epsilon,1\}$ where $\epsilon \in [0,1)$.
	Note that when $\epsilon=0$, the instance is binary.

	In the following, we present an algorithm that computes EFX allocations for three agents with bi-valued cost functions.
	We first give some definitions, some of which will be reused in the next section.
 
	\begin{definition}[Consistent Items]
	    We call item $e\in M$ a \emph{consistently large} item if for all $i\in N, c_i(e) = 1$;
	    a \emph{consistently small} item if for all $i\in N, c_i(e) = \epsilon$.
	    All other items are called \emph{inconsistent} items.
	\end{definition}
	
	\begin{definition}[Large/small Items]
        We call item $e\in M$ \emph{large to agent $i$} if $c_i(e) = 1$;
        \emph{small to $i$} if $c_i(e) = \epsilon$.
        Let $M^-_i = \{ e\in M: c_i(e)=\epsilon \}$ be the items that are small to agent $i$.
        If an item is large (resp. small) to agent $i$, but is small (resp. large) to all other agents, we say that it is \emph{large (resp. small) only to} agent $i$.
    \end{definition}
	    
	The algorithm we use to compute EFX allocations for three agents is based on the Round-Robin algorithm.
	The algorithm (see Algorithm~\ref{alg:RR}) takes as input a set of items and an ordering of the agents $\{ \sigma_1,\sigma_2,\ldots,\sigma_n \} = N$, and lets the agents pick their favourite (minimum cost) item one-by-one, following the order $\sigma$, until all items are allocated.
	We call the output allocation $\bX$ a round-robin allocation.
	Note that in each round, exactly one item is assigned.
    {  We index the rounds by $1,2,\cdots,|M'|$.
	Note that when $M'\neq M$, the allocation is partial.}
	For every agent $i$, we denote the last round during which she receives an item by $r_i$.
    We first show the fairness of round-robin allocation and how we can use the property of round-robin allocation to compute an EFX allocation.
    
	\begin{algorithm}
		\caption{Round-Robin ($M', \sigma_1, \sigma_2, \cdots, \sigma_n$)}
		\label{alg:RR} 
		\KwIn{A set of items $M'\subseteq M$ and an ordering of the agents $(\sigma_1, \sigma_2, \cdots, \sigma_n)$.}
		Initialize $X_i \gets \emptyset$ for all $i\in N$, and let $P\gets M'$, $i\gets 1$\;
		\While{$P\neq \emptyset$}{
		    let $e\gets \argmin_{e\in P}\{c_{\sigma_i}(e)\}$\;
		    update $X_{\sigma_i} \gets X_{\sigma_i} + e$ and $P\gets P-e$\;
		    set $i\gets (i \mod n) + 1$ \;}
 		\KwOut{$\bX=(X_1,X_2,\cdots,X_n)$}
	\end{algorithm}
	
	\begin{lemma}[One-way Envy]{\label{lemma:i-not-envy-j}}
	    {  For a given round-robin allocation of any general additive instance, for any two agents $i\neq j$ such that $r_i<r_j$, agent $i$ does not envy agent $j$, and agent $j$ is EF1 towards agent $i$.}
	\end{lemma}
	\begin{proof}
	    Since $r_i < r_j$, for every round $r$ in which agent $i$ received an item $e$, there exists a round $r'$ satisfying $r<r'<r+n$, in which agent $j$ received an item $e'$.
	    Since $r<r'$, we have $c_i(e)\leq c_i(e')$, which leads to $c_i(X_i) \leq c_i(X_j)$. In other words, agent $i$ does not envy agent $j$.
	    Similarly, the above argument holds for agent $j$ and each round $r\neq r_j$ during which agent $j$ receives an item.
	    Consequently, excluding the last item $e$ agent $j$ receives, we have $c_j(X_j - e)\leq c_j(X_i)$. 
        Hence agent $j$ is EF1 towards agent $i$.
	\end{proof}
	
	\begin{lemma}{\label{lemma:eps-and-1-makes-EFX}}
	    Given the round-robin allocation on items $M'\subseteq M$ and agents $i$ and $j$ with $r_i<r_j$, if there exists $e\in M\setminus M'$ such that $c_i(e)=\epsilon$ and $c_j(e)=1$, then in the allocation $(\ldots,X_i+e,\ldots, X_j,\ldots)$, agent $i$ is EFX towards agent $j$ and agent $j$ does not envy agent $i$.
	\end{lemma}
	\begin{proof}
	    From Lemma~\ref{lemma:i-not-envy-j}, we have $c_i(X_i+e)-\epsilon = c_i(X_i) \leq c_i(X_j)$, which leads to agent $i$ being EFX towards agent $j$.
	    {  As for agent $j$, we have $c_j(X_i+e) = c_j(X_i)+1 \geq c_j(X_j)$ since agent $j$ does not envy agent $i$ by more than one item after the round-robin allocation (before the item $e$ is assigned).
	    Therefore agent $j$ does not envy agent $i$.}
	\end{proof}

    In the following, we use the above lemma to prove the main result in this section.
 
	\begin{theorem}\label{thm:efx-3-agents-bi-valued}
	    There exists an algorithm that computes EFX allocations for three agents with bi-valued cost functions in $O(m \log m)$ time.
	\end{theorem}

    {  As it will be apparent, the running time claimed in Theorem~\ref{thm:efx-3-agents-bi-valued} is dominated by the Round-Robin algorithm, which can be done in $O(m \log m)$ time by sorting items based on each agent's cost function.}
 
	We first note that computing an EFX allocation is easy for some special cases.
	For instance, if there exist two agents sharing the same cost function, e.g., $c_1 = c_2$, we can compute an EFX allocation by partitioning the items into an EFX allocation under the cost function of agent 1, then let agent 3 pick her favorite bundle.
	Clearly, agent 3 would not envy the other two agents.
	In addition, since agents 1 and 2 share the same cost function, no matter which bundle they receive, they are EFX towards the other agents. 
	Hence in the following, we assume that no two agents share the same cost function.
	
	\paragraph{Overview}
	The main idea of our algorithm is to make use of Lemma~\ref{lemma:i-not-envy-j} and~\ref{lemma:eps-and-1-makes-EFX} to construct an EFX allocation.
	In particular, given a round-robin allocation $\bX$, we would like to assign some unallocated items to improve the fairness guarantee to EFX.
    {  Note that every item can be categorized into one of the following four types: \emph{consistently large, consistently small, large only to some agent $i$, small only to some agent $i$}.}
	Suppose every agent $i$ has an item $e_i$ that is small only to agent $i$.
	Then by computing a round-robin allocation on items $M' = M\setminus \{e_1,e_2,e_3\}$ and allocating each $e_i$ to agent $i$, the resulting allocation is EFX, by Lemma~\ref{lemma:eps-and-1-makes-EFX}.
    {  In fact, as we will show in Lemma~\ref{lemma:exist-small-only}, we can find an EFX allocation as long as there is an item that is small only to some agent.}
	When this does not hold, our algorithm carefully examines the number of items of each type in the viewpoint of each agent and proceeds differently.
	
	\begin{lemma} \label{lemma:exist-small-only}
	   If there exists an item that is small only to some agent, then an EFX allocation can be computed in polynomial time.
	\end{lemma}
	\begin{proof}
	    Suppose there exists an agent $i$, say $i=1$, that has an item $e_1$ that is small only to her.
	    Since $c_2(e_1) = c_3(e_1) = 1$, there must exists an item $e_2\neq e_1$ such that $c_2(e_2)\neq c_3(e_2)$, otherwise we have $c_2=c_3$.
	    We can assume w.l.o.g. that $c_2(e_2)=\epsilon, c_3(e_2)=1$.
	    Let $M' = M\setminus \{e_1, e_2\}$.
	    By carefully deciding the ordering of agents, we can compute a round-robin allocation $(S_1,S_2,S_3)$ on items $M'$, satisfying $r_1 < r_2 < r_3$.
	    In other words, agent 3 is the last one that receives an item; agent 2 is the second last one that receives an item.
	    Then we let the final allocation be
	    \begin{equation*}
	        X_1 \gets S_1 + e_1,\quad X_2 \gets S_2 + e_2,\quad X_3 \gets S_3.
	    \end{equation*}
	    
	    From Lemma~\ref{lemma:i-not-envy-j}, we have
        \begin{equation*}
            c_1(S_1)\leq \min\{c_1(S_2), c_1(S_3)\} \quad \text{ and }\quad c_2(S_2)\leq c_2(S_3).
        \end{equation*}
        
	    Since $c_1(e_1) = \epsilon$, agent $1$ certainly is EFX towards the other two agents.
	    Since $c_2(e_2)=\epsilon$, agent $2$ is EFX towards agent $3$.
    	Agent $2$ is EFX towards agent $1$ because $c_2(e_1) = 1$ and thus
        \begin{equation*}
            c_2(X_2-e_2) = c_2(S_2) \leq c_2(S_1+e_1).
        \end{equation*}
    	Finally, by Lemma~\ref{lemma:eps-and-1-makes-EFX}, agent $3$ does not envy the other two agents since $c_3(e_1) = c_3(e_2) = 1$.
    	Hence $(X_1,X_2,X_2)$ is an EFX allocation.
	\end{proof}
    
    By Lemma~\ref{lemma:exist-small-only}, it remains to consider that case when every inconsistent item is large only to some agent.
    In other words, suppose $L_i$ is the set of items that are large only to agent $i$, then $(L_1, L_2, L_3)$ is a partition of the inconsistent items into three sets. We finish the proof of Theorem~\ref{thm:efx-3-agents-bi-valued} by the following two lemmas. 
	
	\begin{lemma} \label{lemma:two-large-only}
	    If there exists $|L_i|\geq 2$, then an EFX allocation can be computed in polynomial time.
	\end{lemma}
	
	\begin{lemma} \label{lemma:one-large-only-each}
	    If $|L_i| \leq 1$ for all $i\in N$, then an EFX allocation can be computed in polynomial time.
	\end{lemma}
	
	In the following, we prove these two lemmas. 
	
	\begin{proofof}{Lemma~\ref{lemma:two-large-only}}
    	W.l.o.g. suppose $|L_1|\geq 2$ and let $e_1,e_2\in L_1$ be any two items that are large only to agent $1$.
    	Note that at least one of $L_2$ and $L_3$ is non-empty, as otherwise we have $c_2 = c_3$.
    	Assume w.l.o.g. that $L_2\neq \emptyset$ and let item $e_3\in L_2$ be large only to agent 2 (see Table~\ref{table:case2}).
    	
    	\begin{table}[htb]
    	\begin{center}
        \begin{tabular}{c|c c c c}
                        & $e_1$         & $e_2$         &$e_3$          &$\cdots$ \\ \hline
        $c_1$           & $1$           & $1$           & $\epsilon$    &$\cdots$    \\ \hline
        $c_2$           & $\epsilon$    & $\epsilon$    & $1$           &$\cdots$    \\ \hline
        $c_3$           & $\epsilon$    & $\epsilon$    & $\epsilon$    &$\cdots$ 
        \end{tabular}
        \end{center}
        \caption{The cost functions of the three agents.}{\label{table:case2}}
        \end{table}
        
        We proceed differently depending on whether $e_3$ is the only small item {  of} agent $1$.
        
        If $M_1^- = \{e_3\}$, i.e., every $e\neq e_3$ is large to agent $1$, then we know that {  agents} 2 and 3 share the same cost function on items $M' = M\setminus\{e_3\}$.
        Because otherwise, we have an item $e'\in M'$ with $c_2(e')\neq c_3(e')$ and $c_1(e')=1$, which leads to $e'$ being small only to either agent 2 or agent 3, which is already handled in Lemma~\ref{lemma:exist-small-only}.
    	To compute an EFX allocation, we divide the items $M'$ into an EFX allocation under the cost function $c_2$, let agent 1 pick her favourite bundle, and then assigns item $e_3$ to agent $1$.
    	The allocation is EFX to agent $2$ and $3$ because they share the same cost function on $M'$; it is also EFX to agent $1$ because $c_1(e_3) = \epsilon$ and $c_1(X_1 - e_3)\leq \min\{ c_1(X_2),c_1(X_3) \}$.
    	
    	If $M_1^-\neq \{e_3\}$, then let $e_4$ be any arbitrary item in $M_1^- \setminus \{e_3\}$.
    	Let $M' = M\setminus\{e_1,e_2,e_3,e_4\}$.
    	By picking the ordering of agents carefully and running the round-robin algorithm, we compute a round-robin allocation $(S_1,S_2,S_3)$ on items $M'$ satisfying that $r_3 < r_2 < r_1$.
    	From Lemma~\ref{lemma:i-not-envy-j}, we have 
    	\begin{equation*}
    	    c_3(S_3)\leq \min\{c_3(S_1), c_3(S_2)\} \quad \text{and}\quad c_2(S_2)\leq c_2(S_1).
    	\end{equation*}
    	
    	We set the final allocation as follows:
    	\begin{equation*}
    	  X_3 \gets S_3+e_2+e_3,\quad X_2 \gets S_2+e_1,\quad X_1 \gets S_1+e_4.
    	\end{equation*}
    	
    	Note that each item in $e_1,e_2,e_3,e_4$ is small to its receiver. The allocation is EFX to
    	\begin{itemize}
    	    \item agent $1$ because $c_1(X_1 - e_4) = c_1(S_1) \leq \min\{c_1(S_2 + e_2 + e_3),c_1(S_3 + e_1)\}$;
    	    \item agent $2$ because $c_2(X_2 - e_1) = c_2(S_2) \leq \min \{c_2(S_1 + e_4),c_2(S_3 + e_3)\}$;
    	    \item agent $3$ because $c_3(X_3 - e_2) = c_3(S_3 + e_3) \leq \min\{c_3(S_1 + e_4), c_3(S_2 + e_1)\}$,
    	\end{itemize}
    	where the inequalities follow from Lemma~\ref{lemma:eps-and-1-makes-EFX}.
	\end{proofof}
	
	\begin{proofof}{Lemma~\ref{lemma:one-large-only-each}}
	    Since $|L_i|\leq 1$ for all $i\in N$, let items $e_1,e_2,e_3$ be large only to agent $1$, $2$, $3$ respectively.
	    Note that some of $e_i$ might be undefined, i.e., when $L_i = \emptyset$.
	    However, since no two agents have the same cost function, at most one of $e_1,e_2,e_3$ is undefined.
	    In the following, we assume w.l.o.g. that $e_1$ and $e_2$ are well defined but not necessarily $e_3$.
        Let $M' = M\setminus\{e_1,e_2,e_3\}$.
        Note that all items in $M'$ are consistent items and thus $c_1,c_2,c_3$ have the same cost values on $M'$.
        We use function $c : M' \rightarrow  \{1,\epsilon\}$ to represent this function.
        We first compute an EFX allocation $\{S_1,S_2,S_3\}$ on $M'$, such that $c(S_1)\leq c(S_2)\leq c(S_3)$.
        Then we try to allocate the remaining items $e_1,e_2,e_3$ to the partial allocation, making sure that the resulting one is EFX.
        
        If $e_3$ is well defined and $c(S_3)-c(S_1)>\epsilon$, we let the final allocation be
        \begin{equation*}
            X_1 \gets S_1+e_2+e_3,\quad
            X_2 \gets S_2+e_1,\quad
            X_3 \gets S_3.
        \end{equation*}
        
        Note that each item in $e_1,e_2,e_3$ is small to its receiver.
        The allocation is EFX because of the following:
        \begin{align*}
            &c_1(X_1-e_2) = c_1(S_1+e_3) \leq \min\{c_1(S_2+e_1),c_1(S_3)\}, \\
            &c_2(X_2-e_1) = c_2(S_2) \leq \min\{c_2(S_1+e_2+e_3,c_2(S_3)\}, \\
            &c_3(X_3-e) = c_3(S_3-e) \leq \min\{c_3(S_1),c_3(S_2)\}, \forall e\in X_3.
        \end{align*}
        
        If $e_3$ is well defined and $c(S_3)-c(S_1)\leq\epsilon$, we let the final allocation be
        \begin{equation*}
            X_1 \gets S_1+e_2,\quad
            X_2 \gets S_2+e_3,\quad
            X_3 \gets S_3+e_1.
        \end{equation*}
        
        Again, each item in $e_1,e_2,e_3$ is small to its receiver.
        The allocation is EFX because of the following:
        \begin{align*}
            &c_1(X_1-e_2)=c_1(S_1)\leq \min\{c_1(S_2+e_3),c_1(S_3+e_1)\}, \\
            &c_2(X_2-e_3)=c_2(S_2)\leq \min\{c_2(S_1+e_2,c_2(S_3)\},\\
            &c_3(X_3-e_1) = c_3(S_3)\leq \min\{c_3(S_1+e_2),c_3(S_2+e_3)\}
        \end{align*}
        
        Finally, if $e_3$ is undefined, we let the final allocation be
        \begin{equation*}
            X_1 \gets S_1+e_2,\quad
            X_2 \gets S_2+e_1,\quad
            X_3 \gets S_3.
        \end{equation*}
        
        As before, each item in $e_1,e_2$ is small to its receiver and the allocation is EFX because:
        \begin{align*}
            &c_1(X_1-e_2) = c_1(S_1)\leq c_1(S_2)\leq c_1(S_3), \\
            &c_2(X_2-e_1) = c_2(S_2)\leq \min\{c_2(S_1+e_2),c_2(S_3)\}, \\
            &c_3(X_3-e) = c_3(S_3-e)\leq \min\{c_3(S_1),c_3(S_2)\}, \forall e\in X_3.
        \end{align*}
        
        Hence for all cases we can compute an EFX allocation $(X_1,X_2,X_3)$.
	\end{proofof}
	
    For completeness, we summarize our complete algorithm in Algorithm~\ref{alg:bivaluedfor3agents}.
    
    \begin{algorithm}[htbp]
		\caption{EFX Allocation for 3 Agents}\label{alg:bivaluedfor3agents} 
		Initialize: $X_1 \gets \emptyset, X_2\gets \emptyset, X_3\gets \emptyset$\;
		\uIf{there exist two agents $i\neq j$ sharing the same cost function}{
		    compute an EFX allocation on item $M$ under $c_i$: $\{S_1,S_2,S_3\}$\;
		    {  for $l$ s.t. $l\neq i$ and $l\neq j$, set $X_l \gets \argmin\{(c_l(S_1), c_l(S_2), c_l(S_3))\}$ \;}
		    arbitrarily allocate the remaining two bundles to agents $i, j$\;}
		\Else(all agents have different cost functions;){
		    \uIf{there exists an item $e_1\in M$ small only to some agent $i$}{
		    there exists $e_2\neq e_1$ s.t. $c_j(e_2) = \epsilon, c_l(e_2) = 1$: let $M^-\gets M\setminus\{e_1,e_2\}$\;
		    {  compute a round-robin allocation on item $M^-$ s.t. $r_i < r_j < r_l:\{S_i,S_j,S_l\}$\;
		    Set $X_i\gets S_i+e_1, X_j\gets S_j+e_2, X_l\gets S_l$\;}}
		    \uElseIf{there exists an agent $i$ that has two items $e_1,e_2$ large only to her}
            {
		    there exists an item $e_3$ that is large only to some agent $j$\;
		    \eIf{there exists an item $e_4\neq e_3$ small to $i$}
            {
		        $M^-\gets M\setminus\{e_1,e_2,e_3,e_4\}$\;
		        {  compute a round-robin allocation on item $M^-$ s.t. $r_l<r_j<r_i:\{S_i,S_j,S_l\}$\;
		        set $X_i\gets S_i+e_4, X_j\gets S_j+e_2, X_l\gets S_l+e_2+e_3$\;}
            }
		    (every $e\neq e_3$ is large to agent $i$;)
            {
            $M^-\gets M\setminus\{e_3\}$\;
		        compute an EFX allocation $S_1,S_2,S_3$ on items $M^-$ under cost function $c_j$\;
		        set $X_i \gets \argmin\{(c_i(S_1), c_i(S_2), c_i(S_3))\}$\;
		        arbitrarily allocate the remaining two bundles to agent $j, l$\;
            }
            }
            \uElseIf{for every agent $i\in N$, there has exactly one item large only to her}
            {
            let $e_1,e_2,e_3$ be the items large only to agent $1,2,3$, $M^-=M\setminus\{e_1,e_2,e_3\}$\;
	        divide $M^-$ into an EFX allocation $\{S_1,S_2,S_3\}$ with $c(S_1)\leq c(S_2) \leq c(S_3)$\;
	        \eIf{$c(S_3)-c(S_1)>\epsilon$}
                {set $X_1\gets S_1+e_2+e_3, X_2\gets S_2+e_1, X_3\gets S_3$\;}
                {set $X_1\gets S_1+e_2, X_2\gets S_2+e_3, X_3\gets S_3+e_1$\;}
            }
            \Else(there are only two items $e_1,e_2$ that is large only to agent $i,j$ respectively;)
            {
            $M^-\gets M\setminus\{e_1,e_2\}$\;
            divide $M^-$ into an EFX allocation $\{S_1,S_2,S_3\}$ with $c(S_1)\leq c(S_2) \leq c(S_3)$\;
            set $X_i\gets S_1+e_2, X_j\gets S_2+e_1,X_l\gets S_3$\;
            }
        }
 	\KwOut{$\bX=(X_1,X_2,X_3)$}
	\end{algorithm}

    \section{Bi-valued Instances with Four or More Agents}\label{sec:bi-fourormore}
	
	In this section, we consider bi-valued instances with $n\geq4$ agents.
	We provide a polynomial-time algorithm that computes an EFX allocation with at most $n-1$ items unallocated; and a polynomial-time algorithm that computes a complete allocation that is $(n-1)$-EFX.

    \begin{theorem}\label{thm:partial-n-bivalued}
	    For the instances of $n\geq 4$ agents with bi-valued cost functions, there exists an algorithm that computes an EFX partial allocation with at most $n-1$ unallocated items.
	\end{theorem}
	
	\begin{theorem}\label{thm:general-n-bivalued}
	    For the instances of $n\geq 4$ agents with bi-valued cost functions, there exists an algorithm that computes an $(n-1)$-EFX allocation.
	\end{theorem}
	
	Given an instance, we divide the set of items into $M^-$ and $M^+$, where $M^+$ includes all consistently large items:
	\begin{equation*}
	    M^- = \{e\in M: \exists i\in N, c_i(e)=\epsilon\}, \quad  M^+ = \{e\in M: \forall i\in N, c_i(e)=1\}.
	\end{equation*}
	
	Note that we have $M^- = \bigcup_{i\in N} M_i^-$, where $M_i^-$ contains the items that are small to agent $i$.
	We first give an algorithm that computes an allocation $\bX^0$ of items $M^-$, for which our results are based on.

	\subsection{Obtaining Partial Allocation and Agent Groups}
	
	We first give an algorithm that computes an allocation $\bX^0$ of items $M^-$ with certain desired properties (see Lemma~\ref{lemma:leximin}).
    When $M^+\neq\emptyset$, $\bX^0$ is a partial allocation.
	Based on the allocation $\bX^0$, we partition agents into disjoint groups $\bA=(A_1,\cdots,A_b)$.
	For the allocation $\bX^0$ and groups $\bA=(A_1,\cdots,A_b)$, we have the following properties.
    We say that agent $i$ is in a \emph{higher} (resp. \emph{lower}) group than $j$ if $i\in A_r$ and $j\in A_{r'}$ for some $r < r'$ (resp. $r > r'$).
    
	\begin{lemma} \label{lemma:leximin}
	    We can compute allocation $\bX^0$ of items $M^-$ to agents $N$, and partition the agents into groups $\bA=(A_1,\cdots,A_b)$ with the following properties:
	    \begin{enumerate}
            \item For all $r\in \{1,2,\ldots,b\}$ and $i,j\in A_r$, we have $\left| |X_i^0|-|X_j^0| \right|\leq 1$. 
            In other words, agents from the same group receive almost the same number of items in $\bX^0$;
            \item For all $i\in A_r$, $j\in A_{r'}$ such that $r < r'$, we have $c_j(X_i^0) = |X_i^0|$.
            In other words, for an agent $j$ in a lower group, all items received by an agent $i$ in a higher group are large;
            \item For all $i\in A_r$, $j\in A_{r'}$ such that $r < r'$, we have $|X_i^0| \geq |X_j^0|$.
            In other words, agent $i$ receives at least as many items as every agent $j$ from a lower group.
            Additionally, we have $\{i\in N:|X_i^0| = \min_{j\in N} |X_j^0| \}\subseteq A_b$.
	    \end{enumerate}
	\end{lemma}
	
    In the following, we present the algorithms to compute allocation $\bX^0$ and partition agents into groups.
    
    \paragraph{The Algorithm}
    Recall that every item $i\in M^-$ is small to some agents.
    We first compute an allocation $\bX = (X_1,\ldots,X_n)$ in which every agent only receives items that are small to her.
    Based on the allocation we construct a directed graph $G$ as follows.
    For every $i,j\in N$, we {  add} an edge from $j$ to $i$ if there exists an item $e\in X_i$ such that $c_j(e) = \epsilon$.
    As long as there exists a path from $j$ to $i$ with $|X_i|-|X_j|\geq 2$, we implement an item-transfer along the path, which increases $|X_j|$ by one and decreases $|X_i|$ by one.
    When there is no path of such type in $G$, we return the allocation as $\bX^0$.
    Note that for all $e\in X^0_i$, we have $c_i(e)=\epsilon$.
    We summarize the construction of $\bX^0$ in Algorithm~\ref{alg:leximin}.
    
	\begin{algorithm}[htbp]
		\caption{Partial Allocation for $M^-$}\label{alg:leximin}
		\KwIn{Fair allocation instance $(N,M^-,C)$ with $c_i(e) \in \{\epsilon,1\}$}
		Initialize: $X_i \gets \emptyset$ for all $i\in N$, $P\gets M^-$\;
		\While{$P\neq \emptyset$}{
		    pick any item $e\in P$ and agent $i\in N$ s.t. $c_i(e)=\epsilon$\;
            update $X_i\gets X_i+e$, $P \gets P-e$;
        }
	    Construct a directed graph $G=(N,E)$: there exists an edge from $j$ to $i$ if $\exists e\in X_i, c_j(e)=\epsilon$\;
	    \While{there exists a path $i_k\rightarrow \cdots \rightarrow i_0$ such that $|X_{i_0}|-|X_{i_k}|\geq 2$}{
	        \For{$l=1,2,\cdots,k$}{
	            {  let $e$ be any item in $X_{i_{l-1}}$ such that $c_{i_l}(e) = \epsilon$}\;
	            {  $X_{i_{l-1}}\gets X_{i_{l-1}} - e$,
	            $X_{i_l}\gets X_{i_l} + e$};}}
 		\KwOut{$\bX^0=(X_1^0,\cdots,X_n^0)$}
	\end{algorithm}

    By re-indexing the agents, we can assume w.l.o.g. that in allocation $\bX^0$ we have
    \begin{equation*}
        |X^0_1| \geq |X^0_2| \geq \cdots \geq |X^0_n|.
    \end{equation*}
    We refer to this property as the \emph{monotonicity} (in size) of the allocation $\bX^0$.
    Based on the construction of item-transfer in line $(6-9)$ of Algorithm~\ref{alg:leximin}, we have the following observations.
    \begin{observation}\label{obser:partial-allocation}
        For the output allocation $\bX^0$, we have:
        \begin{enumerate}
        \item For all $i, j\in N$ s.t. $|X_i^0|-|X_j^0|\geq 2$, for all item $e\in X_i^0$, we have $c_j(e)=1$, i.e., $c_j(X_i^0) = |X_i^0|$;
        \item For all $i,j,l\in N$ s.t. $|X_i^0|-|X_j^0| = |X_j^0|-|X_l^0|=1$, we have either $c_j(X_i^0) = |X_i^0|$ or $c_l(X_j^0) = |X_j^0|$;
        \end{enumerate}
    \end{observation}
	\begin{proof}
	    For agent $i,j\in N$ with $|X_i^0|-|X_j^0|\geq 2$, if there exists an item $e\in X_i$ s.t. $c_j(e)=\epsilon$, then we have a edge from $j$ to $i$ and an item-transfer can be implemented, which is a contradiction.
	    Similarly, if there exist agents $i,j,l\in N$ s.t. $|X_i|-|X_j| = |X_j|-|X_l|=1$ and items $e_1\in S_i, e_2\in S_j$ s.t. $c_j(e_1)=c_l(e_2)=\epsilon$, then there is a path from $l$ to $i$, which would have been resolved in the algorithm.
	\end{proof}
	        
    Next we introduce the partitioning of agent groups.
    We first partition agents into $N_1, \cdots, N_b$ by $|X^0_i|$, such that two agents are in the same group if and only if they have bundles of the same size.
    Specifically, we have $N_1 = \{ i\in N: |X^0_i| = \max_{j\in N} |X^0_j| \}$ contains the agents that receive most items, $N_2 = \{ i\in N: |X^0_i| = \max_{j\in N\setminus N_1} |X^0_j| \}$ contains the agents with second most items, etc.
    Thus 
    \begin{itemize}
        \item for all $r\in \{1,2,\ldots,b\}$ and $i,j \in N_r$, we have $|X^0_i| = |X^0_j|$;
        \item for all $1\leq r < r'\leq b$ and $i \in N_r, j \in N_{r'}$, we have $|X^0_i| > |X^0_j|$.
    \end{itemize}
    
    Observed that if there exists a path from $j\in N_{r+1}$ to $i\in N_r$, then $i$ cannot reach any agent $l\in \cup_{t=1}^{r-1} N_t$, because we have $|X^0_l| \geq |X^0_j| +2$.
    Therefore, all items in $X^0_l$, where $l\in \cup_{t=1}^{r-1} N_t$, are large to agent $i$.
    {  However, we cannot guarantee that all items in $X^0_i$ are large to agent $j$.
    Therefore, to produce the partition of agents $\bA$, we further classify agents in each group $N_{r}$ into two types.
    We divide $N_r$ into $F_r$ and $N_r\setminus F_r$, where each agent $i\in N_r$ is in $F_r$ if it can be reached by a path from some agents in $N_{r+1}$ (see Algorithm~\ref{alg:agent-groups} for a constructive definition of $F_r$).
    For completeness, we let $F_0 = F_b = \emptyset$.}
    Then we define
    \begin{equation*}
        A_r = (N_r\setminus F_r) \cup F_{r-1}, \quad \text{for all } r\in \{1,2,\ldots,b\},
    \end{equation*}
    where $F_0 = F_b = \emptyset$.
    We summarize the construction of $\bA$ in Algorithm~\ref{alg:agent-groups}.

	\begin{algorithm}[htbp]
		\caption{Construct Agent Groups $\bA$}\label{alg:agent-groups}
		\KwIn{Allocation $\bX^0$}
	   initialize: $N' \gets N$, $b\gets 0$\;
        \While{$N' \neq \emptyset$}
        {  
            $b \gets b+1$, $N_b \gets \{ i\in N : |X^0_i| = \max_{j\in N'} |X^0_j| \}$ and $N' \gets N' \setminus N_b$ \;
        }
		for all $r\in \{1,2,\ldots,b-1\}$, let $F_r\gets\{i\in N_r:\exists j\in N_{r+1} \text{ and } e\in X_i, c_j(e) = \epsilon\}$; let $F_0 = F_b = \emptyset$\;
		\While{there exists $i\in N_r\setminus F_r$ such that $\exists j\in F_r\text{ and } e\in X_i, c_j(e) = \epsilon$}{
		    add $i$ to $F_r$;}
		for all $r\in \{1,2,\ldots,b-1\}$, let $A_r\gets (N_r\setminus F_r) \cup F_{r-1}$.   \\
        \KwOut{$\bA=(A_1,\cdots,A_b)$}
	\end{algorithm}

    We present an illustrative example in Figure~\ref{fig:leximin} on the partitioning of agents into groups $\bA$.
    
    \begin{figure}[htb]
    \begin{center}
    \begin{tikzpicture}
    \begin{axis}
        [
        ybar,
        ymax=11,
        enlargelimits=0.1,
        bar width=12pt,
        height=6cm,width=12cm,
        ylabel={Bundle Size},
        symbolic x coords={1, 2, 3, 4, 5, 6, 7, 8, 9, 10, 11, 12}, 
        xtick=data,
        nodes near coords,
        ]
        \addplot+[color = black, fill = gray!70, name nodes near coords=agent]coordinates {(1,10) (2,10) (3,9) (4,9) (5,9) (6,9) (7,8) (8,8) (9,7) (10,3) (11,3) (12,2)};        
    \end{axis}
    \draw [->,black] (agent-2.north) to [out=120,in=15] (agent-1.north);
    \draw [->,black] (agent-5.north) to [out=135,in=45] (agent-4.north);
    \draw [->,black] (agent-6.north) to  [out=120,in=15] (agent-5.north);
    \draw [->,black] (agent-8.north) to [out=120,in=15] (agent-7.north);
    \draw [->,black] (agent-11.north) to [out=120,in=15] (agent-10.north);
    \coordinate (a) at (0.65,-0.3);
    \coordinate (b) at (1.1,-0.3);
    \draw[decorate,decoration={brace,raise=10pt,amplitude=0.15cm},black]
    (b.south) -- (a.south);
    \node[below, black] at (0.9,-0.8) {$A_1$};
    \coordinate (a) at (1.5,-0.3);
    \coordinate (b) at (3.4,-0.3);
    \draw[decorate,decoration={brace,raise=10pt,amplitude=0.15cm},black] (b.south) -- (a.south);
    \node[below, black] at (2.5,-0.8) {$A_2$};
    \coordinate (a) at (3.9,-0.3);
    \coordinate (b) at (5.7,-0.3);
    \draw[decorate,decoration={brace,raise=10pt,amplitude=0.15cm},black] (b.south) -- (a.south);
    \node[below, black] at (4.85,-0.8) {$A_3$};
    \coordinate (a) at (6.25,-0.3);
    \coordinate (b) at (7.3,-0.3);
    \draw[decorate,decoration={brace,raise=10pt,amplitude=0.15cm},black] (b.south) -- (a.south);
    \node[below, black] at (6.8,-0.8) {$A_4$};
    \coordinate (a) at (7.75,-0.3);
    \coordinate (b) at (8.25,-0.3);
    \draw[decorate,decoration={brace,raise=10pt,amplitude=0.15cm},black]
    (b.south) -- (a.south);
    \node[below, black] at (8.05,-0.8) {$A_5$};
    \coordinate (a) at (8.6,-0.3);
    \coordinate (b) at (9.75,-0.3);
    \draw[decorate,decoration={brace,raise=10pt,amplitude=0.15cm},black] (b.south) -- (a.south);
    \node[below, black] at (9.2,-0.8) {$A_6$};
    \end{tikzpicture}        
    \end{center}
    \caption{An illustration of how to partition the agent groups based on the allocation $\bX^0$.
    The arrows in the figure represent that there exists an edge from the tail agent to the head agent.
    As shown in the figure, agent $1$ is in $N_1\setminus F_1 \subseteq A_1$ since there is no edge to her.
    There exists an edge from agent $3$ to agent $2$, which implies that agent $2$ is in $F_1 \subseteq A_2$.
    Note that both agents $3$ and $4$ are in $N_2\setminus F_2 \subseteq A_2$ since they cannot be reached by any agent in $N_3$.
    However, agents $5$ and $6$ are in $F_2\subseteq A_3$ since they can be reached by agent $7$.
    Hence we have $A_3 = \{5,6,7\}$.
    Similarly we have $A_4 = \{8,9\}$, $A_5 = \{10\}$, and $A_6 = \{11,12\}$.}
    \label{fig:leximin}
    \end{figure}
    
	    This allows us to show the properties of groups $\bA=(A_1,\cdots,A_b)$.
	    
        \begin{proofof}{Lemma~\ref{lemma:leximin}}
        We first show property $1$.
        Recall that group $A_r=F_{r-1}\cup (N_r\setminus F_r)$.
        Hence we have $\left| |X_i^0|-|X_j^0| \right|\leq 1$ for all $i,j\in A_r$.
	    Property $2$ follows from Observation~\ref{obser:partial-allocation} and the construction of $A_r$ as follows.
     
        Consider any $i\in A_r, j\in A_{r'}$, where $r < r'$.
        \begin{itemize}
            \item If $i\in F_{r-1}, j\in F_{r'-1}$: by definition agent $j\in F_{r'-1}$ has a path from some agents in $N_{r'}$.
            By Observation~\ref{obser:partial-allocation}, for all $e\in X_i$, we have $c_j(e)=1$.

            \item If $i\in F_{r-1}, j\in N_{r'}\setminus F_{r'}$:
            since $i\in N_{r-1}$ and $r' \geq (r-1) + 2$, by Observation~\ref{obser:partial-allocation}, for all $e\in X_j$, we have $c_i(e)=1$.
            
            \item If $i\in N_r\setminus F_r, j\in F_{r'-1}$: 
            if $r' - 1 = r$, then by definition for all $e\in X_i$, $c_j(e)=1$, otherwise agent $i$ should be add into $F_{r}$.
            Otherwise $r' \geq r + 2$ and since agent $j\in F_{r'-1}$ has a path from some agents in $N_{r'}$, by Observation~\ref{obser:partial-allocation}, for all $e\in X_i$, we have $c_j(e)=1$.
        
            \item If $i\in N_r\setminus F_r, j\in N_{r'}\setminus F_{r'}$:
            if $r' = r + 1$, by definition agent $i$ can not be reached by a path from any agent in $N_{r+1} = N_{r'}$.
            Hence for all $e\in X_i$, we have $c_j(e)=1$.
            Otherwise $r' - r \geq 2$. By Observation~\ref{obser:partial-allocation}, for all $e\in X_j$, we have $c_i(e)=1$.
        \end{itemize}

        Finally, property $3$ follows from the construction of $A_b$.
        Since $N_b$ contains agents $i$ with minimum $|X_i|$, we have $F_b=\emptyset$, $A_b=F_{b-1}\cup (N_b\setminus F_b) = F_{b+1}\cup N_b$.
        Hence we have $N_b\subseteq A_b$.
	\end{proofof}

    \subsection{EFX Partial Allocation with at most \texorpdfstring{$n-1$}{} Unallocated Items}
    
    In this section, we present an algorithm that computes an EFX allocation with at most $n-1$ unallocated items in polynomial time.
%
    Our algorithm begins by calling Algorithm~\ref{alg:leximin} and~\ref{alg:agent-groups} to return a partial allocation $\bX^0$.
    Recall that at this moment the unallocated items are $P = M^+$, i.e., the consistently large items.
    Our goal is to allocate items in $P$, and also reallocate some items in allocation $\bX^0$, so that all agents receive roughly the same number of items, until $|P| < n$.
    By Lemma~\ref{lemma:leximin}, we can guarantee EFX among agents within the same group as the number of items they receive differ by at most one, and all agents receive only small items.  
    To maintain the EFX-ness among agents within the same group, during the reallocation we ensure that 
    \begin{itemize}
        \item[(1)] all new items allocated to an agent are large to her and her groupmates (agents within the same group);
        \item[(2)] all agents from the same group receive the same number of new items.
    \end{itemize}
    
    To further ensure EFX-ness between agents from different groups, we require that in the final allocation
    \begin{itemize}
        \item[(3)] for all agents $i,j\in N$, we have $\left| |X_i|-|X_j| \right| \leq 1$.
    \end{itemize}
    
	\paragraph{The Algorithm}
    As introduced, the goal of our algorithm is to decide the allocation of $P$ and reallocation of items in $\bX^0$ to satisfy the three constraints listed above.
    For all $r\in \{ 1,2,\ldots,b \}$, let $k_r = \min_{i\in A_r}\{ |X^0_i| \}$.
    By Lemma~\ref{lemma:leximin}, for all $i\in A_r$ we have $|X^0_i| \in \{ k_r + 1, k_r \}$.
    We aim to compute an integer $\delta_r$ (which can be negative) for each group $A_r$, such that $k_r + \delta_r = k$ for all $r\in \{1,2,\ldots, b\}$, where $k$ is an appropriately chosen parameter.
    This ensures that constraint~(3) is satisfied.
    Then for each $r$, if $\delta_r$ is negative, we remove $|\delta_r|$ items from $X^0_i$, for each $i\in A_r$; if $\delta_r$ is positive, we include $\delta_r$ items to $X^0_i$.
    The above operations ensure that constraints (1) and (2) are satisfied.
    Finally, by setting
    \begin{equation}
        k = \left\lfloor \frac{\sum_{r=1}^b (k_r\cdot |A_r|)+|M^+|}{n} \right\rfloor, \label{eq:k-definition}
    \end{equation}
    we can show that at most $n-1$ items are left unallocated, and the allocation is EFX.
    
    The steps of the full algorithm are summarized in Algorithm~\ref{alg:partial-ngeq4-bivalued}.

    \begin{algorithm}[htbp]
		\caption{Partial EFX Allocation for $n\geq 4$ Agents}\label{alg:partial-ngeq4-bivalued} 
		Initialize: $X_i \gets \emptyset$ for all $i\in N$, $P\gets M^+$ \;
		compute partial allocation $\bX^0$ on items $M^-$ and agent groups $\bA$, using Algorithm~\ref{alg:leximin} and~\ref{alg:agent-groups}\;
        for all $r\in \{1,2,\ldots,b\}$, let $k_r = \min_{i\in A_r}\{ |X^0_i| \}$ and define $k = \left\lfloor \frac{\sum_{r=1}^b (k_r\cdot |A_r|)+|M^+|}{n} \right\rfloor$ \;
		\For{$r = 1,2,\ldots, b$}
		{
            define $\delta_r \gets k - k_r$ \;
            \If{$\delta_r < 0$}
            {
                \For{each $i\in A_r$}
                {
                    pick arbitrary $|\delta_r|$ items $\Delta_i$ from $X^0_i$, and update $X_i \gets X^0_i \setminus \Delta_i$, $P \gets P \cup \Delta_i$ \; 
                }
            }
            \Else
            {
                \For{each $i\in A_r$}
                {
                    pick arbitrary $\delta_r$ items $\Delta_i$ from $P$, and update $X_i \gets X^0_i \cup \Delta_i$, $P \gets P \setminus \Delta_i$ \; 
                }
            }
		}
        \KwOut{$\bX=(X_1,\cdots,X_n)$}
	\end{algorithm}    
    
    \begin{lemma}
        {  Algorithm~\ref{alg:partial-ngeq4-bivalued} computes an EFX allocation with at most $n-1$ unallocated items in $O(n m^2)$ time.}
    \end{lemma}
    \begin{proof}
        By definition we have $k_1 \geq k_2 \geq \cdots \geq k_b$.
        Therefore in Algorithm~\ref{alg:partial-ngeq4-bivalued},
        \begin{itemize}
            \item there exists $1\leq r_1 < r_2 \leq b$ such that agents in $\cup_{r\leq r_1} A_r$ lose items and agents in $\cup_{r\geq r_2} A_r$ receive items.
            By Lemma~\ref{lemma:leximin}, items from agents in higher groups are large to agents in lower groups, constraint~(1) is satisfied.
            In fact, we have the stronger property that every new item allocated to an agent is large to all agents in $\cup_{r\geq r_2} A_r$;
            
            \item for all $r\in \{1,2,\ldots,b\}$, each agent $i\in A_r$ receives or loses $|\delta_i|$ items. 
            Hence constraint~(2) is satisfied;
            
            \item in the final allocation, for all $r\in \{1,2,\ldots,b\}$, each $i\in A_r$ has $|X_i| \in \{k+1,k\}$.
            Hence constraint~(3) is satisfied;
        \end{itemize}
        
        In the following, we use the above properties to show that the allocation is EFX.
        Trivially, agents within the same group are EFX towards each other, as their costs differ by at most $\epsilon$.
        Next, consider any two agents $i\in A_r$ and $j\in A_{r'}$ from different groups, where $r < r'$.
        \begin{itemize}
            \item Since $i$ is from a higher group, we have $c_j(X_i) = |X_i|\geq k \geq c_j(X_j - e)$ for all $e\in X_j$.
            Hence agent $j$ is EFX towards agent $i$.
            \item {  Since $k_r \geq k_{r'}$, we have $\delta_r \leq \delta_{r'}$. 
            If agent $i$ receives new items, then the number of new items agent $j$ receives is not smaller.
            Moreover, these new items are large to both $i$ and $j$.
            If agent $i$ does not receive new items, then all items in $X_i$ are small to agent $i$.
            Since by removing any item $e\in X_i$, $|X_i - e|\leq |X_j|$, for both cases we have $c_i(X_i - e)\leq c_i(X_j)$.
            Hence agent $i$ is EFX towards agent $j$;}
        \end{itemize}
        
        Next, we show that there are at most $n-1$ unallocated items, i.e., $|P| \leq n-1$ at the end of the algorithm.
        Since initially $|P| = |M^+|$ and each agent $i\in A_r$ decreases $|P|$ by $\delta_r$, it suffices to show that $|M^+| - \sum_{r=1}^{b} (\delta_r\cdot |A_r|) \leq n-1$, which is true because
        \begin{align*}
            |M^+| - \sum_{r=1}^{b} (\delta_r\cdot |A_r|) & = |M^+| - \sum_{r=1}^{b} \left((k-k_r)\cdot |A_r|\right)
             = |M^+| + \sum_{r=1}^{b} (k_r\cdot |A_r|) - k\cdot n \\
            & = \left( |M^+| + \sum_{r=1}^{b} (k_r\cdot |A_r|) \right) - \left\lfloor \frac{\sum_{r=1}^b (k_r\cdot |A_r|)+|M^+|}{n} \right\rfloor \cdot n \leq n-1.
        \end{align*}
        
        {  Finally, we prove that Algorithm~\ref{alg:partial-ngeq4-bivalued} runs in polynomial time.
        We first show that both Algorithms~\ref{alg:leximin} and~\ref{alg:agent-groups} run in polynomial time.
        Observe that 1) the computation of the initial allocation in the while loop in lines $2-4$ in Algorithm 6 takes $O(nm)$ time; 2) given the initial allocation, computing the directed graph $G$ takes $O(nm)$ time since each agent has to recognize at most $m$ edges (for convenience we maintain a multi-graph).
        It suffices to argue that the while loops in lines $6-8$ finish in polynomial time.
        In the following, we show that by carefully choosing the paths in the while loop in lines $6-8$, the while loops break after $O(m)$ rounds.
        When there are multiple satisfied paths, we choose the one that maximizes $|X_{i_0}|$, which can be identified in $O(nm)$ time by running depth-first searches (there are $n$ nodes and at most $nm$ edges in the directed graph).
        Following such path selection principle, we can guarantee that once an agent is selected as the end of a path, she will not be selected as the start of any future paths.
        Therefore throughout the whole algorithm, the items reallocated from the end agents are all different, which implies that there are $O(m)$ rounds.
        Therefore Algorithm~\ref{alg:leximin} runs in $O(n m^2)$ time.

        In Algorithm~\ref{alg:agent-groups}, the computation of groups $N_r$ can be done in $O(n)$ time.
        The construction of groups $A$ replies on the partition of each $N_r$ into $F_r$ and $N_r \setminus F_r$, which can be done by recognizing the connected component of each agent in $N_r$ in $O(m)$ time.
        Note that the number of connected components is $O(n)$. Thus Algorithm~\ref{alg:agent-groups} runs in $O(nm)$ time.
        
        Finally, we consider the remaining Algorithm~\ref{alg:partial-ngeq4-bivalued}.
        Note that the computation of $\{k_r, \delta_r\}_{r\leq b}$ and $k$ take $O(m)$ time.
        Each item incurs $O(1)$ time in the reallocation procedure since each item is reallocated at most once.
        In conclusion, Algorithm~\ref{alg:partial-ngeq4-bivalued} runs in $O(n m^2)$ time.}
    \end{proof}
    
    We present an illustrative example in Figure~\ref{fig:partial-allocation} that shows how to compute an EFX partial allocation based on the partial allocation $\bX^0$ and $P$ with agents groups $\bA$ in Figure~\ref{fig:leximin}.
    \begin{figure}[htb]
    \begin{center}
    \begin{tikzpicture}
    \begin{axis}
        [
        ybar stacked,
        ymax=11,
        enlargelimits=0.1,
        bar width=12pt,
        height=6cm,width=12cm,
        ylabel={Bundle Size},
        symbolic x coords={1, 2, 3, 4, 5, 6, 7, 8, 9, 10, 11, 12}, 
        xtick=data,
        nodes near coords,
        ]
        \addplot+[color = black, fill = gray!70, name nodes near coords=agent]coordinates {(1,7) (2,8) (3,7) (4,7) (5,8) (6,8) (7,7) (8,8) (9,7) (10,3) (11,3) (12,2)};
        \addplot+[color = black, dashed, fill = white, name nodes near coords=agent]coordinates {(1,3) (2,2) (3,2) (4,2) (5,1) (6,1) (7,1) (8,0) (9,0) (10,0) (11,0) (12,0)};   
        \addplot+[color = black, fill = gray!10, postaction={
        pattern=north east lines}, name nodes near coords=agent]coordinates {(1,0) (2,0) (3,0) (4,0) (5,0) (6,0) (7,0) (8,0) (9,0) (10,4) (11,5) (12,5)};  
    \end{axis}
    \coordinate (a) at (0.65,-0.3);
    \coordinate (b) at (1.1,-0.3);
    \draw[decorate,decoration={brace,raise=10pt,amplitude=0.15cm},black] (b.south) -- (a.south);
    \node[below, black] at (0.9,-0.8) {$A_1$};
    \coordinate (a) at (1.5,-0.3);
    \coordinate (b) at (3.4,-0.3);
    \draw[decorate,decoration={brace,raise=10pt,amplitude=0.15cm},black] (b.south) -- (a.south);
    \node[below, black] at (2.5,-0.8) {$A_2$};
    \coordinate (a) at (3.9,-0.3);
    \coordinate (b) at (5.7,-0.3);
    \draw[decorate,decoration={brace,raise=10pt,amplitude=0.15cm},black] (b.south) -- (a.south);
    \node[below, black] at (4.85,-0.8) {$A_3$};
    \coordinate (a) at (6.25,-0.3);
    \coordinate (b) at (7.3,-0.3);
    \draw[decorate,decoration={brace,raise=10pt,amplitude=0.15cm},black] (b.south) -- (a.south);
    \node[below, black] at (6.8,-0.8) {$A_4$};
    \coordinate (a) at (7.75,-0.3);
    \coordinate (b) at (8.25,-0.3);
    \draw[decorate,decoration={brace,raise=10pt,amplitude=0.15cm},black] (b.south) -- (a.south);
    \node[below, black] at (8.05,-0.8) {$A_5$};
    \coordinate (a) at (8.6,-0.3);
    \coordinate (b) at (9.75,-0.3);
    \draw[decorate,decoration={brace,raise=10pt,amplitude=0.15cm},black] (b.south) -- (a.south);
    \node[below, black] at (9.2,-0.8) {$A_6$};
    \end{tikzpicture}        
    \end{center}
    \caption{An illustration on how to compute an EFX partial allocation based on $\bX^0$ with $|P| = 2$.
    As shown in the figure, $k_1 = 10, k_2=9, k_3=8, k_4=7, k_5=3, k_6=2$.
    Hence we have $k = 7$.
    Agents in $A_1, A_2, A_3$ hold that $\delta_1 = -3, \delta_2 = -2, \delta_3 = -1$, and the dashed areas represent the removed items.
    For agents $8, 9$ we have $\delta_4 = 0$, which means that they do not lose or receive any item during the algorithm.
    Finally for agents in $A_5, A_6$ we have $\delta_5 = 4, \delta_6 = 5$, the hatching areas represent the item agents $10, 11, 12$ receive.
    In the EFX partial allocation we have $|X_i| \in \{7, 8\}$ for all $i\in N$.}
    \label{fig:partial-allocation}
    \end{figure}

    \subsection{\texorpdfstring{$(n-1)$}{}-EFX Allocation}
    
    In this section, we show that there exists an algorithm that computes an $(n-1)$-EFX allocation in polynomial time.
    Recall from Lemma~\ref{lemma:m-n+1} that we can compute $(m-n)$-EFX allocations for every instance with $m$ items and $n$ agents, which gives an $(n-1)$-EFX allocation when $m \leq 2n-1$.
    In the following, we assume that $m \geq 2n$.
    Since we can compute an EFX partial allocation with $n-1$ unallocated items, the natural idea is to find a way to allocate the unallocated items so that the approximation ratio does not increase by too much.
    In fact, if $\delta_b > 0$, i.e., agents in the lowest group receive new items, then by allocating all unallocated items to the agent $i\in A_b$ with minimum $|X_i|$, it can be proved that the allocation is $(n-1)$-EFX.
    Unfortunately, as we will see in this section, the most subtle case is when $\delta_b = 0$. 
    For example, after collecting items from higher groups, we still have $|P| < |A_b|$, and thus to ensure that all agents from the same group receive the same number of new items, we cannot allocate any new item to any agent.
    To handle this issue, we propose the \textit{Small Item Reallocation} (SIR) algorithm (a formal description will be given later).

    \paragraph{The Algorithm}
    Our algorithm is based on the EFX partial allocation computation we have introduced in the previous section.
    Let $\bX^1 = \{X_1^1, \cdots, X_n^1\}$ be the EFX partial allocation with at most $n-1$ unallocated items.
    In the following, we allocate the unallocated items $P$ to agents in $A_b$, i.e., the lowest group.
    The detailed steps are summarized in Algorithm~\ref{alg:ngeq4-bivalued}.
    We first attempt to allocate the items in a round-robin manner to agents in $A_b$ in the order of $(n,n-1,\ldots)$. 
    We show that the resulting allocation is $(n-1)$-EFX if any of the following three conditions holds: (1) $\delta_b \geq 1$; (2) $|P| \geq |A_b|$; (3) every agent in $A_b$ who does not receive new items has cost at least $1/(n-2)$.
    If none of the above conditions holds, then (instead of using Round-Robin) we introduce a method called \textit{Small Item Reallocation} (SIR), which scans through agents in $A_b$ one by one.
    For each scanned agent $i\in A_b$, {  if it holds an item that costs $1$ to all unscanned agents, then we do not allocate any further item to this agent; otherwise, we reallocate all items in its bundle, making sure that each item is small to its receiver, and allocate the bundle an item from $P$.}
    The scanning stops when $|P| = 0$ or $|P|$ equals the number of unscanned agents, in the latter case we allocate one item in $P$ to each unscanned agent.

    \smallskip
    
    \begin{algorithm}[htbp]
		\caption{$(n-1)$-EFX for $n\geq 4$ Agents}\label{alg:ngeq4-bivalued} 
		initialize: $X_i \gets \emptyset$ for all $i\in N$\;
		compute partial allocation $\bX^1$ with unallocated items $P$, and agents groups $\textbf{A}$ using Algorithm~\ref{alg:partial-ngeq4-bivalued}\;
        let $p \gets |P|$ and $s \gets |A_b|$, i.e., we have $A_b = \{ n,n-1,\ldots,n-s+1 \}$\;
		\eIf{$\delta_b > 0$ or $s \leq p$ or $c_{n-p}(X^1_{n-p})\geq {1}/{(n-2)}$}
		{
            allocate items in $P$ to agents in $A_b$ in a round robin manner, in the order of $(n,n-1,\ldots)$.
        }
        {
            \For{$i = n, n-1, \cdots, n-s+1$}{
                \If{$\forall e\in X^1_i, \exists j < i, c_j(e)=\epsilon$}
                {
                    \While{$X_i\neq\emptyset$}
                    {
                    pick an item $e\in X_i$ and agent $j<i$ with $c_j(e) = \epsilon$ and update
                    $X_i\gets X_i-e, X_j\gets X_j+e$\;
                    }
                    pick an arbitrary item $e\in P$ and update $X_i\gets \{e\}, P\gets P-e$\;
                    \If{$P = \emptyset$}{
                    break the for-loop\;
                    }
                }
                \If{$|P| = s-(n-i+1)$}{
                    allocate one item in $P$ to each agent in $\{i-1,i-2,\ldots,n-s+1\}$ and break the for-loop.
                    }
            }
        }
        \KwOut{$\bX=(X_1,\cdots,X_n)$}
	\end{algorithm}
    
    Thus (to prove Theorem~\ref{thm:general-n-bivalued}) it suffices to show that all agents are $(n-1)$-EFX towards the other agents.
    We first claim that for any instance that $m \geq 2n$, we have $k \geq 1$.
    Recall that for any $i\in A_r$, we have $|X_i^0| \in \{k_r, k_r +1\}$.
    Hence from Equation~\eqref{eq:k-definition} we have
        \begin{equation*}
            k = \left\lfloor \frac{\sum_{r=1}^b (k_r\cdot |A_r|)+|M^+|}{n} \right\rfloor \geq \left\lfloor \frac{\sum_{i=1}^n (|X_i^0|)+|M^+| - n}{n} \right\rfloor \geq  \left\lfloor \frac{m - n}{n} \right\rfloor \geq  \left\lfloor \frac{2n-n}{n} \right\rfloor \geq 1.
        \end{equation*}

    Note that every item $e\in P$ is either a consistently large item, i.e. $e\in M^+$, or taken from some bundle $X_i$ such that $i \in A_r$ with $\delta_r < 0$.
    By definition of $k$ we have $\delta_b \geq 0$, which implies that all items in $P$ are either from $M^+$ or groups higher than $A_b$.
    From Lemma~\ref{lemma:leximin}, we conclude that all items in $P$ are large to all agents in $A_b$, i.e., $\forall e\in P, i\in A_b$, we have $c_i(e) = 1$.
%
    We first show some conditions under which Round-Robin (with ordering $(n, n-1, \cdots, n-s+1)$) computes an $(n-1)$-EFX allocation.
    By the following lemma, we know that if line 5 of Algorithm~\ref{alg:ngeq4-bivalued} is executed then the resulting allocation is $(n-1)$-EFX.
    
    \begin{lemma}\label{lemma:properties}
    If allocating items in $P$ to agents in $A_b$ in a round robin manner with ordering $n, n-1, \cdots, n-s+1$ does not give an $(n-1)$-EFX allocation, then we must have (1) $\delta_r = 0$; (2) $|A_b| > |P|$; and (3) $c_{n-p}(X_{n-p}^1) < \frac{1}{n-2}$.
    \end{lemma}
    \begin{proof}
        We show that if any of the listed conditions does not hold, then the Round-Robin algorithm returns an $(n-1)$-EFX allocation.
        Recall that $\bX^1$ is an EFX partial allocation.
        Let $\bX = (X_1,\ldots, X_n)$ be the allocation returned by Round-Robin.
        Since we have $X^1_j\subseteq X_j$ for all agent $j\in N$, any agent $i$ that does not receive any new item certainly is $(n-1)$-EFX towards any other agents.
        Hence it suffices to argue that any agent $i$ that receives new items is $(n-1)$-EFX towards any other agents.
        Fix any such agent $i$.
        We have $i\in A_b$ since we only allocate new items to the last group.
        {  Moreover, since each item in $X_i\setminus X^1_i$ is large to $i$, when arguing $(n-1)$-EFX, we can assume w.l.o.g. that the item to be removed from $X_i$ is actually in $X^1_i$, as the newly added items are large to agent $i$.}
        Let $c = \max_{e\in X^1_i} c_i(X^1_i - e)$, it suffices to prove that for all $j\neq i$, 
        \begin{equation}
            c + |X_i \setminus X^1_i| \leq (n-1)\cdot c_i(X_j). \label{eqn:efx-by-rr}
        \end{equation}

        Recall from the definition of $k$, we have $|X^1_j|\in \{k, k+1\}$ for all $j\in N$.
        Therefore
        \begin{itemize}
            \item for all $j\notin A_b$ we have $c_i(X_j) \geq k$;
            \item for all $j \in A_b\setminus\{i\}$ we have $c_i(X_j) \geq c$ (since $\bX^1$ is EFX).
        \end{itemize}
                
        \paragraph{Suppose $\delta_b > 0$ or $|A_b| \leq |P|$}
        If $\delta_b > 0$, then each agent in $A_b$ (including $i$) receives at least one new item in Algorithm~\ref{alg:partial-ngeq4-bivalued} (line $9-11$).
        If $|A_b| \leq |P|$, then each agent $A_b$ receives at least one item from $P$ during the Round-Robin allocation.
        All these items are large to agent $i$.
        Hence in either {  case} we have $c_i(X_j) \geq 1$ for all $j\in N$.
        If $|A_b| = 1$, then we have $|X^1_i| = k$ and $c\leq k-1$. Therefore for all $j\neq i$ we have $j\notin A_b$ and thus
        \begin{equation*}
            c + |X_i\setminus X_i^1| \leq k-1 + n-1 \leq (n-1) \cdot k \leq (n-1)\cdot c_i(X_j).
        \end{equation*}
        If $|A_b| \geq 2$, then for all $j\notin A_b$ we have
        \begin{equation*}
            c + |X_i \setminus X^1_i| \leq k + \left\lceil \frac{p}{|A_b|} \right\rceil \leq k + \left\lceil \frac{n}{2} \right\rceil \leq \left(\left\lceil \frac{n}{2} \right\rceil+1\right)\cdot c_i(X_j) \leq (n-1)\cdot c_i(X_j).
        \end{equation*}
        For all $j\in A_b\setminus\{i\}$ we have $c_i(X_j) \geq 1$, which implies (recall that we also have $c_i(X_j) \geq c$)
        \begin{equation*}
            c + |X_i \setminus X^1_i| \leq c_i(X_j) + \left\lceil \frac{n}{2} \right\rceil \leq \left(\left\lceil \frac{n}{2} \right\rceil + 1\right)\cdot c_i(X_j) \leq (n-1)\cdot c_i(X_j).
        \end{equation*}
        
        \paragraph{Suppose $\delta_b = 0$, $|A_b| > |P|$ and $c_{n-p}(X_{n-p}^1) \geq \frac{1}{n-2}$}
        By definition, each agent in $\{n,n-1,\ldots,n-p+1\}$ receives one item from $P$, where $p = |P| < |A_b| = s$, and the agents in $\{n-p,\ldots,n-s+1\}$ do not receive any new item. 
        Therefore we have $i\in \{n,n-1,\ldots,n-p+1\}$ and $|X_i\setminus X^1_i| = 1$.
        Hence for all $j\notin A_b$, we have
        \begin{equation*}
            c + |X_i \setminus X^1_i| \leq c_i(X_j) + 1 \leq  2\cdot c_i(X_j).
        \end{equation*}
        For all $j\in A_b\setminus\{i\}$, if $j\in \{n,n-1,\ldots,n-p+1\}$ then $c_i(X_j) \geq 1$; if $j\in \{n-p,\ldots,n-s+1\}$ then
        \begin{equation*}
            c_i(X_j) = c_i(X^1_j) \geq c_j(X^1_j) \geq c_{n-p}(X^1_{n-p}) \geq \frac{1}{n-2},
        \end{equation*}
        where the first inequality holds because all items in $X^1_j$ that are large to $j$ must also be large to all agents in $A_b$; the second inequality holds because of the monotonicity of $|X^1_j|$ (for all $j\in A_b$).
        In either case, we have $c_i(X_j) \geq 1/(n-2)$, which implies (recall that we also have $c_i(X_j) \geq c$)
        \begin{equation*}
            c + |X_i \setminus X^1_i| \leq c_i(X_j) + 1 \leq  (1 + (n-2))\cdot c_i(X_j) = (n-1)\cdot c_i(X_j).
        \end{equation*}
        
        Hence if any of the three conditions does not hold, Round-Robin returns an $(n-1)$-EFX allocation.
    \end{proof}

    Next, we show that when all three properties listed in Lemma~\ref{lemma:properties} hold, SIR computes an $(n-1)$-EFX allocation.

    \begin{lemma}\label{lemma:SIR}
    If we have $\delta_r = 0$, $|A_b| > |P|$ and $c_{n-p}(X_{n-p}^1) < \frac{1}{n-2}$, then SIR computes an $(n-1)$-EFX allocation.
    \end{lemma}
    \begin{proof}
    Note that since $|X^1_i|\in \{k,k+1\}$ and $c_{n-p}(X^1_{n-p}) < 1/(n-2)$, we have $k\cdot \epsilon < 1/(n-2)$.
    Recall that in the SIR method, we scan the agents one by one and for each scanned agent $i$ we either replace its bundle with a single item from $P$ (in line $9 - 13$) or do not allocate any item from $P$ to $i$.
    We first argue that all items in $P$ are allocated in the final allocation.
    Recall that $|A_b| > |P|$.
    Since in each for-loop the number of unscanned agents decreases by one, and $|P|$ decreases by at most one, the algorithm must reach a state with $|P| \in \{0,s-(n-i+1)\}$, where $s - (n-i+1)$ is the number of unscanned agents.
    In either case ($|P| = 0$ or $|P| = s-(n-i+1)$) the algorithm will terminate with $P=\emptyset$ and each agent receiving at most one item from $P$.
    Therefore in the final allocation $\bX = (X_1,\ldots,X_n)$, for all $i\in A_b$,
    \begin{itemize}
        \item[(a)] $|X_i| = 1$ and $c_j(X_i) = 1$ for all $j\in A_b$; or
        \item[(b)] $X^1_i\subseteq X_i$. Moreover, in this case we have $c_i(X_i) = \epsilon\cdot |X_i|$ (recall that $\delta_b=0$ and we reallocate items among agent only if the item is small to its receiver, e.g., in line $10$).
    \end{itemize}
    
    Fix any agent $i\in A_b$, we argue that $i$ is $(n-1)$-EFX towards any other agent $j$, which is trivially true when $|X_i| = 1$.
    Therefore it suffices to consider that agent $i$ falls into the case (b).
    By the above property, for all $j\in A_b$ we have either $c_i(X_j) = 1$ or $X^1_j \subseteq X_j$.
    Moreover, in the latter case if $j > i$ (agent $j$ is scanned before agent $i$), then we also have $c_i(X_j) \geq 1$.
    If $c_i(X_j) \geq 1$, then $i$ is $(n-1)$-EFX towards $j$ because for all $e\in X_i$,
    \begin{align*}
		c_i(X_i - e) & \leq ((k+1)\cdot |A_b| - 2)\cdot \epsilon \leq (kn + n - 2)\cdot \epsilon \\
		& \leq (kn + k(n - 2))\cdot \epsilon = (n-1)\cdot 2k\cdot \epsilon \leq n-1,
	\end{align*}
    where the first inequality holds because there exists at least one bundle $X^1_j$ in $A_b$ of size $k$. 
    
    Otherwise $i > j$, which implies that $i$ is scanned before agent $j$ and $i\geq 2$.
    Hence all items in $X_i$ are from bundles $X^1_n, X^1_{n-1},\ldots, X^1_{i}$, all these bundles have size at most $|X^1_j|$ (by monotonicity of bundle sizes in $\bX^1$).
    Hence $c_i(X_i) \leq (n-i+1)\cdot c_i(X_j) \leq (n-1)\cdot c_i(X_j)$.
    \end{proof}

    \begin{proofof}{Theorem~\ref{thm:general-n-bivalued}}
        We show that Algorithm~\ref{alg:ngeq4-bivalued} computes an $(n-1)$-EFX (full) allocation in polynomial time.
        {  By Lemma~\ref{lemma:properties}, if the condition in line 4 of Algorithm~\ref{alg:ngeq4-bivalued} holds, then Round-Robin computes an $(n-1)$-EFX allocation in $O(n)$ time since there are $|P|$ rounds and in each round the agent picks an arbitrary item since they are large to all agents in $A_b$.}
        Otherwise by Lemma~\ref{lemma:SIR}, SIR computes an $(n-1)$-EFX allocation.
        Moreover, since it takes $O(n\cdot |X^1_i|)$ time to process each scanned agent $i$, the running time $O(nm)$ is also polynomial. 
    \end{proofof}
    
    \section{Conclusion and Future Work}\label{sec:conclusion}
	
	In this paper, we propose algorithms that always compute a $(2+\sqrt{6})$-EFX allocation for three agents and $(3n^2-n)$-EFX allocation for $n\geq 4$ agents. 
	These are the first approximation ratios of EFX that are independent of $m$ for the allocation of indivisible chores.
	Furthermore, we show that the approximation ratios can be improved for bi-valued instances.
	We propose algorithms that always compute an EFX allocation for three agents with bi-valued cost functions and an EFX allocation with at most $n-1$ items unallocated for $n\geq 4$ agents with bi-valued cost functions.
	For the bi-valued instances with $n\geq 4 $ agents, we further propose an algorithm that computes a complete allocation with an approximation ratio $n-1$ with respect to EFX.
	There are many open problems regarding the computation of approximately EFX allocation for chores.
	For example, it would be interesting to investigate whether constant approximations of EFX allocation exist for general number of agents, and whether EFX allocations exist for three agents or any number of agents with bi-valued cost functions.
	Observe that to ensure (approximation of) EFX for an agent $i$, we often need to focus on increasing the costs of other agents, instead of minimizing the cost of $i$, which can possibly lead to inefficiency in the final allocation.
	It is thus interesting to study the existence of allocations for chores that are fair, e.g., approximation of EFX, MMS, or PROPX, and efficient, e.g., Pareto optimal.
	

\bibliographystyle{abbrv}
\bibliography{efx}

\end{document}